\documentclass[11pt]{amsart}

\usepackage[T1]{fontenc}
\usepackage{microtype}      

\usepackage{amsmath,amssymb,amsfonts,amsthm}
\allowdisplaybreaks

\usepackage{mathtools}
\usepackage{mathrsfs}
\usepackage{bm}

\usepackage{placeins}

\numberwithin{equation}{section}

\usepackage[margin=1in]{geometry}

\usepackage{graphicx}
\usepackage{enumitem}

\usepackage[numbers,sort&compress]{natbib}

\usepackage[pdftex,hyperfootnotes=false]{hyperref}
\hypersetup{
	colorlinks=true,     
	pdfstartview=FitH,
	bookmarksopen=true,
	bookmarksnumbered=true
}

\theoremstyle{plain}
\newtheorem{theorem}{Theorem}[section]
\newtheorem{theorem*}{Theorem}
\newtheorem{maintheorem}{Theorem}

\newtheorem{mainproposition}[maintheorem]{Proposition}

\newtheorem{proposition}[theorem]{Proposition}
\newtheorem{lemma}[theorem]{Lemma}
\newtheorem{corollary}[theorem]{Corollary}

\theoremstyle{definition}
\newtheorem{definition}[theorem]{Definition}

\theoremstyle{remark}
\newtheorem{remark}[theorem]{Remark}

\newcommand{\CC}{\mathbb{C}}
\newcommand{\QQ}{\mathbb{Q}}
\newcommand{\R}{\mathbb{R}}
\newcommand{\NN}{\mathbb{N}}

\DeclareMathOperator{\Disc}{Disc}

\DeclareMathOperator{\Ran}{Ran}
\DeclareMathOperator{\Span}{Span}

\newcommand{\HS}{\mathrm{HS}}

\begin{document}

\title{Spectral Structure of the Mixed Hessian of the Dispersionless Toda $\tau$-Function}

\author{Oleg Alekseev}

\address{Department of Mathematics, HSE University (National Research University Higher School of Economics), Moscow, Russia}

\address{Chebyshev Laboratory, Department of Mathematics and Computer Science, Saint Petersburg State University, Saint Petersburg, Russia}



\date{}

\begin{abstract}
We study the mixed Hessian of the dispersionless Toda \(\tau\)-function for
the one-harmonic \(s\)-fold symmetric conformal map
\(f(w)=rw+aw^{1-s}\). This Hessian is the susceptibility matrix generated by
the inverse conformal map. Our spectral statements are formulated for its
weighted symmetry-block realizations on a fixed Hilbert space. In that
realization, the first spectral transition occurs at the analytic threshold
\(\zeta_c=(s-1)^{s-1}/s^s\), where the dominant square-root singularity of the
inverse map reaches the normalization circle, rather than at the geometric
threshold \(\zeta_{\mathrm{univ}}=1/(s-1)\), where univalence fails. After
symmetry decomposition and weighted realization, each block develops exactly
one logarithmically diverging eigenvalue as \(\zeta\uparrow\zeta_c\), while
the remaining spectrum stays bounded and converges to a compact limit. The
instability is therefore rank one in every symmetry sector of the weighted
block theory.

We then continue the scalar Gram generating functions beyond \(\zeta_c\).
They are generalized hypergeometric functions on the slit plane
\(\mathbb{C}\setminus[\zeta_c^2,\infty)\), their branch-point expansion contains
the logarithmic term responsible for the divergence, and in the range
\(1\le p\le s\) they admit Cauchy--Stieltjes and Jacobi-matrix realizations.
In particular, the continued scalar quantities remain regular at
\(\zeta_{\mathrm{univ}}\), so the analytic spectral transition strictly
precedes the geometric breakdown of univalence.
\end{abstract}

\maketitle


\section{Introduction}
\label{sec:intro}


The dispersionless Toda hierarchy associates to a planar domain a
\(\tau\)-function whose mixed second derivatives with respect to harmonic
moments define a natural susceptibility matrix. For simply connected domains,
this mixed Hessian is generated by the inverse conformal map and therefore
encodes analytic information about the domain in a form adapted to
integrable-hierarchy methods
\cite{MineevWeinstein2000,WiegmannZabrodin2000,MarshakovWiegmannZabrodin2002,BOOKtau,Takhtajan2001,GustafssonLin2013,ZabrodinWiegmann2006,Teo2009}.
In the present paper we study this Hessian for the explicit one-harmonic
\(s\)-fold symmetric family $
	f(w)=rw+aw^{1-s}$,
which provides a nontrivial solvable case at the interface of conformal
mapping, dispersionless integrability, and spectral analysis.

Our main observation is that the first spectral instability of the weighted
symmetry-block realization occurs strictly before the conformal map ceases to be univalent. More precisely, we identify
an analytic threshold $\zeta_c=(s-1)^{s-1}/s^s$, at which the dominant
square-root singularity of the inverse map reaches the normalization
circle, and we prove that this threshold is strictly smaller than the
geometric threshold $\zeta_{\mathrm{univ}}=1/(s-1)$ for loss of univalence.
The spectral statement is formulated for the weighted block realizations on a
fixed Hilbert space introduced below: in that realization, each symmetry
block develops exactly one logarithmically diverging eigenvalue as
\(\zeta\uparrow\zeta_c\), while all remaining eigenvalues stay bounded.
Thus the onset of instability is rank one in every symmetry sector.

Let \(D\subset\CC\) be a bounded simply connected domain with rectifiable
boundary, and let
\[
	f:\{|w|>1\}\to\CC\setminus D,
	\qquad
	f(w)=rw+a_0+a_1w^{-1}+\cdots,
\]
be the unique exterior conformal map normalized by \(f(w)/w\to r\) as
\(w\to\infty\). Writing \(w=w(z)\) for the inverse map near infinity, one
obtains the mixed logarithmic kernel
\begin{equation}\label{eq:intro-grunsky}
	\log\!\left(1-\frac{1}{w(z)\,\overline{w(z')}}\right)
	=
	-\sum_{m,n\ge1}H_{mn}\,
	\frac{(z/r)^{-m}}{m}\frac{(\bar z'/r)^{-n}}{n},
\end{equation}
valid for \(|z|\) and \(|z'|\) sufficiently large. It is obtained by pulling
back the canonical mixed logarithmic kernel of the exterior disk under the
inverse conformal map. In the dispersionless Toda framework this is precisely
the Wiegmann--Zabrodin mixed kernel identity
\cite{WiegmannZabrodin2000,BOOKtau}. The coefficients \(H_{mn}\) are the main
object of the present paper.

In the geometric realization of the dispersionless \(2\)D Toda hierarchy,
the exterior conformal map evolves under commuting flows whose times
\((t_0,t_1,t_2,\dots;\bar t_1,\bar t_2,\dots)\) are
identified with the harmonic moments of \(\CC\setminus D\). In a standard normalization, \(t_0=(1/\pi)\int_D d^2z\) is the area of \(D\)
divided by \(\pi\), while for \(n\ge1\) the variables \(t_n\) and \(\bar t_n\)
are the holomorphic and antiholomorphic harmonic moments of
\(\CC\setminus D\). The hierarchy
possesses a dispersionless \(\tau\)-function whose logarithm
\(\mathcal F=\log\tau\) plays the role of a free energy, and in a standard
normalization one has
\[
	H_{mn}
	=
	r^{-(m+n)}\frac{\partial^2\mathcal F}{\partial t_m\,\partial\bar t_n}.
\]
Thus \((H_{mn})\) is the scale-invariant mixed Hessian of the dispersionless
free energy. This interpretation is closely connected with Takhtajan's
free-boson approach to the \(\tau\)-function for smooth Jordan contours
\cite{Takhtajan2001}. In the genus-zero case, \(\log\tau\) is a K\"ahler
potential for a natural Hermitian metric on the space of contours of fixed
area, so the mixed derivatives of \(\log\tau\) represent the corresponding
metric coefficients in harmonic-moment coordinates.

The same coefficients therefore appear naturally in several standard
frameworks. In
Laplacian growth and Hele--Shaw dynamics, \((H_{mn})\) describes the
second-order coupling between moment deformations. In the dispersionless Toda
hierarchy, it is the matrix of mixed second derivatives of \(\log\tau\). In
Bergman-kernel language, one has, with our normalization,
\begin{equation}\label{eq:bergman-identity}
	\pi B_{\CC\setminus D}(z,z')
	=
	-\partial_z\partial_{\bar z'}
	\log\!\left(1-\frac{1}{w(z)\,\overline{w(z')}}\right),
\end{equation}
so the same coefficients may be recovered from the exterior Bergman kernel in
conformal coordinates \cite{Bell1992,Jeong2007}. In the normal matrix model,
\(\mathcal F\) is the planar free energy, or equivalently the logarithmic
energy of the equilibrium measure on \(D\)
\cite{ZabrodinWiegmann2006,Teo2009}, and \((H_{mn})\) may be viewed as a
mixed \emph{susceptibility} matrix with respect to moment deformations.

The main contribution of the paper is to identify the first singularity seen
by the block spectrum of the mixed Hessian and to show that this singularity
is analytic rather than geometric. This question is genuinely different from
the classical holomorphic Grunsky theory, whose operator bounds are tied
directly to univalence and quasiconformal extension
\cite{Pommerenke1992,Duren1983}. In the mixed sector studied here, the first instability already appears at the
analytic threshold and concentrates into one logarithmically diverging
\emph{stiff} mode per symmetry sector, while the remaining sectorial spectrum
stays bounded. The resulting
separation between analytic and geometric criticality is invisible in a
purely holomorphic framework. It should also be distinguished from spectral
and determinantal questions for the classical Grunsky operator, including
recent asymptotics for truncated holomorphic Grunsky operators
\cite{Schiffer1981,Shen2010,JohanssonViklund2023}: those works concern the
holomorphic Grunsky coefficients or Fredholm determinants built from them, whereas
here the basic object is the mixed Hessian \((H_{mn})\), and the relevant
transition is a logarithmic rank-one spike at the earlier analytic threshold
\(\zeta_c\).

We now restrict to a family in which the analysis is completely explicit. We
study the \(s\)-fold symmetric one-harmonic family
\begin{equation}\label{eq:map}
	f(w)=rw+aw^{1-s},
	\qquad
	s\ge2,\quad r>0,\quad a\in\R,
\end{equation}
which is the simplest family in which the analytic singularity of the inverse
map and the geometric loss of univalence separate. Maps of the form
\eqref{eq:map} are standard in Laplacian growth and related free-boundary
problems \cite{GustafssonLin2013}. They already exhibit nontrivial geometric
transitions, including loss of univalence and, for \(s\ge3\), boundary cusps
at the geometric threshold. Because they depend on a single dimensionless
shape parameter \(\zeta:=a/r>0\), they provide a minimal and fully explicit
example in which the spectral behavior of the mixed Hessian can be analyzed in
detail. The present paper is deliberately restricted to this one-harmonic
family. Its purpose is to isolate the threshold mechanism and the resulting
rank-one spectral instability in the simplest case where both can be made
completely explicit.

Writing $x:=r/z$ and $w(z)^{-1}=x\,U(x;\zeta)$,
one finds that the inverse map is governed by the algebraic equation
\begin{equation}\label{eq:intro-U}
	U=1+\zeta x^sU^s.
\end{equation}
This elementary relation already contains the two thresholds that control the
problem. The first is the \emph{analytic threshold}
\[
	\zeta_c=\frac{(s-1)^{s-1}}{s^s},
\]
at which the dominant square-root singularity of the inverse map reaches
the unit circle in the \(x\)-plane. The second is the \emph{geometric threshold}
\[
	\zeta_{\mathrm{univ}}=\frac{1}{s-1},
\]
at which the conformal map ceases to be univalent and the boundary develops
its first geometric singularity: for \(s\ge3\), semicubical cusps, and for
\(s=2\), the degenerate Joukowski segment. The problem is to determine which
of these thresholds governs the first spectral instability seen by the
weighted block realizations of the mixed Hessian.

For the weighted block Hessian, the relevant threshold is~\(\zeta_c\). The \(s\)-fold
symmetry of \eqref{eq:map} implies \(H_{mn}=0\) unless
\(m\equiv n\pmod s\), so the Hessian decomposes into \(s\) independent
blocks. These blocks admit an exact positive Gram factorization in terms of
Raney coefficients generated by \eqref{eq:intro-U}. The spectral theorem is
proved for the associated \emph{weighted block realizations} on the fixed
Hilbert space \(\ell^2(\NN_0)\): the coefficient matrix is intrinsic, while
the weighting is the method that makes the logarithmic instability visible in
a fixed operator framework. In this realization, each symmetry block develops
exactly one logarithmically diverging eigenvalue, while the remaining
sectorial spectrum stays bounded. Equivalently, there is one stiff mode per
symmetry sector, so at most \(s\) logarithmically diverging eigenvalues
globally. The first spectral transition is therefore an analytic phenomenon
of the inverse map rather than a manifestation of geometric
non-univalence.

A second part of the paper concerns the intermediate regime, $\zeta_c<\zeta<\zeta_{\mathrm{univ}}$,
where the conformal map remains
univalent although the weighted compact operator realization from the
subcritical regime is no longer available. The block-operator framework
therefore disappears, but the scalar objects underlying the Gram
construction remain meaningful and continue analytically. More precisely, the
Gram weights of the subcritical theory are generated by series built from the
squared Raney coefficients, and these scalar functions admit analytic
continuation beyond \(\zeta_c\). We identify the continued functions as
generalized hypergeometric objects, derive their resonant branch-point
expansion, establish a Cauchy--Stieltjes representation, and, for \(1\le p\le s\), obtain a Jacobi-matrix realization. Thus the continued
scalar Gram quantities retain the same analytic branch-point mechanism that
underlies the subcritical logarithmic instability, even though the weighted
compact realization is no longer available. In particular, the singularity at \(\zeta_c\) is not a geometric one: the two lateral boundary values of the
continued scalar quantities remain finite at the univalence threshold.

\subsection{Main results}
\label{subsec:main-results}

The main results address three points. Theorem~\ref{thm:A} gives the
operator-theoretic form of the analytic threshold \(\zeta_c\): for a fixed
weight parameter \(\beta>0\), each symmetry block splits into one
logarithmically stiff direction and a bounded soft remainder. Theorem~\ref{thm:B}
identifies the scalar Gram quantities that survive beyond the weighted compact
regime and continues them across the cut as generalized hypergeometric and
Stieltjes objects. Proposition~\ref{prop:C} then shows that this analytic transition
occurs strictly before the geometric loss of univalence at
\(\zeta_{\mathrm{univ}}\).

Fix a symmetry sector \(q\in\{1,\dots,s\}\) and a weight parameter
\(\beta>0\). Let \(H^{(q)}\) denote the intrinsic coefficient matrix of the
\(q\)-th Hessian block, and let \(\widetilde G^{(q)}(\zeta)\) be the
corresponding \emph{weighted} block Gram operator. The spectral analysis in this
paper is carried out for these weighted realizations on the fixed Hilbert
space \(\ell^2(\NN_0)\), where \(\NN_0:=\{0,1,2,\dots\}\). This point is
important: the coefficient matrix \(H^{(q)}\) is intrinsic, but the spectral
theorem concerns the weighted operators \(\widetilde G^{(q)}(\zeta)\) and
equivalently \(\widetilde H^{(q)}(\zeta)\). The weighting therefore does not modify the coefficients. Rather, it places the problem in a fixed Hilbert-space framework in which the critical rank-one spike can be isolated. Since the weighted Gram and weighted Hessian blocks have the same nonzero spectrum, the spectral conclusions may be stated equivalently for \(\widetilde G^{(q)}(\zeta)\) or \(\widetilde H^{(q)}(\zeta)\).

\begin{maintheorem}[Rank-one logarithmic instability]
	\label{thm:A}
	Fix \(q\in\{1,\dots,s\}\) and \(\beta>0\). The weighted block Gram
	operator admits, as \(\zeta\uparrow\zeta_c\), a rank-one logarithmic
	decomposition
	\[
		\widetilde G^{(q)}(\zeta)
		=
		L(\zeta)\,\widetilde{\bm d}^{(q)}\otimes \widetilde{\bm d}^{(q)*}
		+\widetilde C^{(q)}(\zeta),
		\qquad
		L(\zeta):=\log\Bigl(\frac{1}{1-\zeta^2/\zeta_c^2}\Bigr)\to+\infty,
	\]
	where \(\widetilde{\bm d}^{(q)}\neq0\) is independent of \(\zeta\), while
	\(\widetilde C^{(q)}(\zeta)\) remains uniformly bounded and converges in
	operator norm to a compact limit.

	In particular, exactly one eigenvalue in the \(q\)-th sector diverges
	logarithmically,
	\[
		\mu_1^{(q)}(\zeta)
		=
		\|\widetilde{\bm d}^{(q)}\|_{\ell^2}^2\,L(\zeta)+O(1).
	\]
	All remaining sectorial eigenvalues stay bounded. More precisely, after
	removing the spike direction, each fixed soft eigenvalue converges to the
	corresponding eigenvalue of the compressed compact limit described in
	Proposition~\ref{prop:soft-spectrum-convergence}.
\end{maintheorem}

Theorem~\ref{thm:A} is proved by combining the square-root singularity analysis of the
inverse map with the blockwise Gram representation of the Hessian. The
logarithmic term is first isolated at the level of matrix entries, and the
weighted realization for fixed \(\beta>0\) then yields a compact operator
framework in which the rank-one singular part separates cleanly from the
bounded soft sector.

The second principal result concerns the scalar quantities that underlie the
block Gram construction itself and therefore remain meaningful beyond the
weighted operator regime. For each \(p\ge1\), consider the generating
function
\begin{equation}\label{eq:intro-Gp}
	\mathcal G_p(u)
	:=
	\sum_{m\ge0}R_{s,p}(m)^2\,u^m.
\end{equation}
For the subcritical conformal family one has the specialization \(u=\zeta^2\).
Below \(\zeta_c\), the \emph{Gram weights} entering the block operators are
obtained from these functions by explicit Euler-type differential operators.
The importance of \(\mathcal G_p\) is therefore that, in scalar form, they
capture the same branch-point mechanism that drives the logarithmic spectral
instability. These functions continue analytically beyond the radius of
convergence \(u=\zeta_c^2\), even though the weighted compact operator
picture no longer persists there.

\begin{maintheorem}[Analytic continuation of the generating function]
	\label{thm:B}
	For each integer \(p\ge1\), the function \(\mathcal G_p\) extends from its disk
	of convergence to a single-valued holomorphic function on
	\(\CC\setminus[\zeta_c^2,\infty)\). It is a generalized hypergeometric
	function, and near the branch point \(u=\zeta_c^2\) it admits a resonant
	expansion of the form
	\begin{equation}\label{eq:intro-resonant}
		\mathcal G_p(u)
		=
		A(u)+B(u)\Bigl(1-\frac{u}{\zeta_c^2}\Bigr)^2
		\log\Bigl(1-\frac{u}{\zeta_c^2}\Bigr),
	\end{equation}
	with \(A\) and \(B\) analytic near \(\zeta_c^2\). The scalar Gram weights are recovered from \(\mathcal G_p\) by explicit
Euler-type differential operators. These remove the quadratic prefactor and
produce the logarithmic divergence at \(u=\zeta_c^2\).
\end{maintheorem}

Theorem~\ref{thm:B} is obtained from the explicit coefficient ratio of
\(R_{s,p}(m)^2\), which identifies \(\mathcal G_p\) with a generalized
hypergeometric function. The local expansion
\eqref{eq:intro-resonant} is then derived by Frobenius analysis at the
regular singular point \(u=\zeta_c^2\). Moreover, the boundary values across
the cut define a discontinuity density, and \(\mathcal G_p\) admits a
Cauchy--Stieltjes representation. In the positive range \(1\le p\le s\), the corresponding Stieltjes measure is
positive by the Hausdorff-moment theorem of Liu--Pego
\cite{LiuPego2016,LiuPegoCorrigendum2026}. As a consequence, \(\mathcal G_p\)
admits a Jacobi-matrix realization as a Weyl \(m\)-function.

The third principal result clarifies the relation between analytic and
geometric criticality.

\begin{mainproposition}[Separation of thresholds]
	\label{prop:C}
	For every \(s\ge2\),
	\begin{equation}\label{eq:intro-separation}
		\zeta_c<\zeta_{\mathrm{univ}}.
	\end{equation}
	Hence the logarithmic spectral instability of Theorem~\ref{thm:A} in the weighted block realization occurs
	while the conformal map is still univalent and the boundary is still
	smooth. Moreover, the analytically continued scalar quantities remain finite at
	\(\zeta=\zeta_{\mathrm{univ}}\)
	\textup{(Propositions~\ref{prop:separation} and~\ref{prop:univ-regularity})}.
\end{mainproposition}

Thus the singularity seen in the weighted mixed-Hessian block theory at \(\zeta_c\) is not a consequence of
cusp formation. It is already encoded in the branch-point structure of the
inverse map and therefore precedes geometric breakdown.

\subsection{Organization of the paper}
\label{subsec:organization}

Section~\ref{sec:inversion} analyzes the inverse conformal map, derives the
Raney-number formulas, and identifies the square-root singularity at the
analytic threshold. Section~\ref{sec:Hessian} introduces the scale-invariant
Toda Hessian, proves the blockwise Gram factorization, and presents the scalar
Gram quantities that will later be continued across the threshold.
Section~\ref{sec:regime-I} develops the weighted spectral theory in the
subcritical regime \(0<\zeta<\zeta_c\) and proves Theorem~\ref{thm:A}.
Section~\ref{sec:regime-2} studies the scalar continuation problem beyond
\(\zeta_c\), including the hypergeometric representation, the resonant
expansion, the Cauchy--Stieltjes formula, and the Jacobi realization.
Finally, Section~\ref{sec:univalence} treats the geometric threshold
\(\zeta_{\mathrm{univ}}\), proves the strict separation
\eqref{eq:intro-separation}, and shows that the continued scalar quantities remain
finite at the univalence boundary. The appendices contain the uniform Raney
asymptotics, the entrywise rank-one extraction, the hypergeometric
continuation argument, and the computation of the branch-point coefficient.


\section{Inversion of the conformal map}
\label{sec:inversion}


We begin with the inverse conformal map, whose singularity structure drives
the later spectral analysis. The main results of this section are the functional equation for the inverse map
(Proposition~\ref{prop:functional-eq}), the explicit Taylor coefficients in terms of Raney numbers
(Proposition~\ref{prop:raney-lagrange}), and the identification of the dominant singularity together with its square-root character
(Lemma~\ref{lem:U-critical}). The last point produces the universal
$m^{-3/2}$ coefficient decay responsible for the logarithmic spectral
divergence established in Section~\ref{sec:regime-I}.

Throughout this manuscript we introduce the scale-free variable
\begin{equation}\label{eq:x-def}
	x := \frac{r}{z}.
\end{equation}

\subsection{The generating function and its functional equation}
\label{subsec:generating-function}

We represent the inverse map $w(z) \sim z/r$ as $z \to \infty$ by factoring out
the leading linear growth:
\begin{equation}\label{eq:w-U}
	w(z) = \frac{z}{r} \, U(x;\zeta)^{-1},
\end{equation}
where $U$ is a scalar function to be determined. The normalization
$w(z) \sim z/r$ corresponds to $U(0;\zeta) = 1$.

\begin{proposition}[Functional equation for the inverse map]
	\label{prop:functional-eq}
	Let $f(w) = rw + aw^{1-s}$ with $\zeta = a/r$, and let $w(z)$ be the inverse
	map normalized by $w(z) \sim z/r$ at infinity. Define $U(x;\zeta)$ by
	\eqref{eq:w-U}. Then $U$ satisfies the algebraic functional equation
	\begin{equation}\label{eq:U-equation}
		U(x;\zeta) = 1 + \zeta x^s U(x;\zeta)^s, \qquad U(0;\zeta) = 1.
	\end{equation}
\end{proposition}

\begin{proof}
	Substituting \eqref{eq:w-U} into the defining relation $z = rw + aw^{1-s}$ and
	using $x = r/z$ gives
	\[
		1 = U^{-1} + \zeta x^s U^{s-1}.
	\]
	Multiplying by $U$ yields \eqref{eq:U-equation}.
\end{proof}

Throughout we restrict to the one-parameter family with $\zeta:=a/r>0$.
For the present one-parameter family, changing the sign of $a$ amounts to a
rigid rotation of the image domain by $\pi/s$.

\subsection{Square-root criticality}
\label{subsec:sqrt-critical}

We now determine the local behavior of $U(x;\zeta)$ at its dominant singularity. This is the central analytic ingredient for
the spectral analysis: the square-root nature of the singularity produces the
universal $m^{-3/2}$ coefficient asymptotics that drive the logarithmic
divergence of the Hessian.

\begin{lemma}[Critical point and local expansion]
	\label{lem:U-critical}
	The function $U(x;\zeta)$ defined by $U = 1 + \zeta x^s U^s$ is a power
	series in the combination $\zeta x^s$, with radius of convergence
	$\zeta_c = (s-1)^{s-1}/s^s$ in that variable. Equivalently, for fixed
	$\zeta>0$ its radius of convergence in $x$ is $(\zeta_c/\zeta)^{1/s}$.
	Its only singularity in the variable $t:=\zeta x^s$ on $|t|=\zeta_c$ is at $t=\zeta_c$.
	Equivalently, for fixed $\zeta>0$ the corresponding singularities in the
	$x$-plane form the $s$-point orbit $\zeta x^s=\zeta_c$.
	At this point the function has a square-root branch point with the local
	expansion
	\begin{equation}\label{eq:U-sqrt}
		U(x;\zeta) = \frac{s}{s-1} - \kappa \sqrt{1 - \frac{\zeta x^s}{\zeta_c}}
		+ O\!\left(1 - \frac{\zeta x^s}{\zeta_c}\right), \qquad \zeta x^s \to \zeta_c,
	\end{equation}
	where $\kappa = \sqrt{2s/(s-1)^3} > 0$.
\end{lemma}

\begin{proof}
	Consider the implicit function $F(y,x;\zeta) := y - 1 - \zeta x^s y^s$. The Taylor branch
	$U(x;\zeta)$ is defined by $F(U,x;\zeta) = 0$ with $U(0;\zeta) = 1$.
	A singularity of this branch occurs where the implicit function theorem fails,
	i.e., where $\partial_y F = 0$:
	\[
		F(U,x;\zeta) = 0, \qquad \partial_y F(U,x;\zeta) = 1 - s\zeta x^s U^{s-1} = 0.
	\]
	From the second equation, $\zeta x^s U^{s-1} = 1/s$. Substituting into the
	first equation:
	\[
		U = 1 + \zeta x^s\,U^s = 1 + \frac{U}{s},
	\]
	which gives $U_c : = s/(s-1)$. Substituting back:
	\[
		\zeta_c = \frac{1}{s\,U_c^{s-1}}
		= \frac{1}{s} \left(\frac{s-1}{s}\right)^{s-1}
		= \frac{(s-1)^{s-1}}{s^s}.
	\]
	Since $U$ is algebraic, its singularities occur only at discriminant points
	of the polynomial $F(y,x;\zeta)=y-1-\zeta x^s y^s$ in the variable $y$.
	Equivalently, they occur exactly when the system
	\[
		F(y,x;\zeta)=0,
		\qquad
		\partial_yF(y,x;\zeta)=0
	\]
	has a common root. The computation above shows that the only nonzero
	critical value of the product $\zeta x^s$ is $\zeta_c$.

	To determine the local behavior, we verify that the singularity is simple
	(i.e., $\partial_{yy} F \neq 0$ at the critical point):
	\[
		\partial_{yy} F = -s(s-1)\zeta x^s U^{s-2}
		= -\frac{(s-1)}{U_c} = -\frac{(s-1)^2}{s} \neq 0.
	\]
	Moreover, near any critical point with $\zeta x^s=\zeta_c$ and $y=U_c$, the
	Taylor expansion of $F$ has the form
	\[
		F(y,x;\zeta)
		=
		\frac12\,\partial_{yy}F(U_c,x;\zeta)\,(y-U_c)^2
		- U_c^s\,(\zeta x^s-\zeta_c)
		+ O\!\bigl(|y-U_c|^3 + |\zeta x^s-\zeta_c|\,|y-U_c|\bigr),
	\]
	because $\partial_yF(U_c,x;\zeta)=0$ at criticality and the coefficient of
	$(\zeta x^s-\zeta_c)$ is $-U_c^s$. Since
	$\partial_{yy}F(U_c,x;\zeta)\neq0$, the Weierstrass preparation theorem gives
	a square-root branch point, hence \eqref{eq:U-sqrt}. Matching the leading
	coefficients yields
	\[
		\frac12\,\partial_{yy}F(U_c,x;\zeta)\,\kappa^2
		+
		U_c^s\zeta_c
		=
		0.
	\]
	Using $\partial_{yy}F(U_c,x;\zeta)=-(s-1)^2/s$ and
	$U_c^s\zeta_c=1/(s-1)$, we obtain
	\[
		\kappa^2=\frac{2s}{(s-1)^3},
	\]
	as claimed.
\end{proof}

\begin{remark}[Geometry in the $z$-plane]
	\label{rem:z-plane}
	Returning to the original variable $z$ via $x = r/z$, the condition
	$|\zeta x^s| = \zeta_c$ becomes $|z| = r(\zeta/\zeta_c)^{1/s}$. Thus the $s$ branch
	points of the inverse map $w(z)$ lie on a circle whose radius shrinks toward $r$
	as $\zeta \uparrow \zeta_c$. When $\zeta = \zeta_c$, the branch points reach
	the circle $|z| = r$, which is the image of the unit circle under the leading
	term $z = rw$ of the conformal map. This is the geometric manifestation of the
	analytic threshold.
\end{remark}

\subsection{Raney numbers and coefficient formulas}
\label{subsec:raney}

The Taylor coefficients of powers of $U$ admit a classical closed form. We
define these coefficients directly and then identify them with the Raney numbers. For any $p \in \mathbb{Z}$, write
\begin{equation}\label{eq:raney-generating}
	U(x;\zeta)^{p} = \sum_{n=0}^{\infty} R_{s,p}(n) \, (\zeta x^s)^n.
\end{equation}

\begin{proposition}[Raney numbers]
	\label{prop:raney-lagrange}
	For all integers $s \ge 2$, $p \ge 1$, and $n \ge 0$,
	\begin{equation}\label{eq:raney-closed}
		R_{s,p}(n) = \frac{p}{sn + p} \binom{sn + p}{n}.
	\end{equation}
	For $p \ge 1$, these are the classical Raney numbers. In particular, the case $p = 1$
	gives the Fuss--Catalan numbers \cite{Raney1960,GrahamKnuthPatashnik1994}.
\end{proposition}

\begin{proof}
	Set $U (t) = 1 + V(t)$ with $V(0) = 0$. Then
	$V = t(1+V)^s$, which is the Lagrange form $V = t\Phi(V)$ with
	$\Phi(v) = (1+v)^s$. We apply the Lagrange--B\"urmann inversion formula to the
	composite function $F(V) = (1+V)^{p} = U^{p}$. The standard
	coefficient extraction \cite{Henrici1974,Comtet1974} gives, for $n \ge 1$,
	\[
		[t^n] F(V(t)) = \frac{1}{n} [v^{n-1}] F'(v) \Phi(v)^n.
	\]
	Since $F'(v) = p(1+v)^{p-1}$ and $\Phi(v)^n = (1+v)^{sn}$, we obtain
	\[
		[t^n] U(t)^{p}
		= \frac{p}{n} [v^{n-1}] (1+v)^{sn+p-1}
		= \frac{p}{n} \binom{sn+p-1}{n-1}.
	\]
	The identity
	$\frac{p}{n}\binom{sn+p-1}{n-1} = \frac{p}{sn+p}\binom{sn+p}{n}$
	yields \eqref{eq:raney-closed} for $n \ge 1$. For $n = 0$, we have
	$R_{s,p}(0) = U(0)^{p} = 1$ (when $p \neq 0$), which agrees
	with \eqref{eq:raney-closed}.
\end{proof}

\begin{remark}[Extension to $p \le 0$]
	\label{rem:raney-negative}
	The generating function $U(x;\zeta)^{p}$ is well-defined for all
	$p \in \mathbb{Z}$. For $p = 0$, we have $U^0 \equiv 1$,
	so $R_{s,0}(n) = \delta_{n,0}$. For $p < 0$, the formula
	\eqref{eq:raney-closed} can be extended using the generalized binomial
	coefficient, but requires separate treatment at values of $n$ where
	$sn + p = 0$. Since all applications in this paper involve $p \ge 1$
	(the index $p$ in the Hessian expansion), we restrict to positive $p$
	in the main statement.
\end{remark}

\begin{corollary}[Convolution property]
	\label{cor:convolution}
	For any integers $k\ge1$ and $p_1,\ldots,p_k\in\mathbb{Z}$, and any $m\ge0$,
	\begin{equation}\label{eq:convolution}
		\sum_{\substack{n_1 + \cdots + n_k = m \\ n_i \ge 0}}
		\prod_{i=1}^{k} R_{s,p_i}(n_i) = R_{s,p_1 + \cdots + p_k}(m).
	\end{equation}
\end{corollary}

\begin{proof}
	Multiply the generating series \eqref{eq:raney-generating} and extract the
	coefficient of $(\zeta x^s)^m$.
\end{proof}

We now present the expansions needed for the Hessian kernel in
Section~\ref{sec:Hessian}. The kernel involves the logarithm
$\log(1 - (w\bar{w}')^{-1})$, which we will expand in powers of $w^{-1}$. Hence
we need the following.

\begin{lemma}[Powers of the inverse map]
	\label{lem:inverse-powers}
	For each integer $p \ge 1$,
	\begin{equation}\label{eq:w-inverse-p}
		w(z)^{-p} = \left(\frac{r}{z}\right)^p \sum_{m=0}^{\infty} R_{s,p}(m) \,
		(\zeta x^s)^m
		= \sum_{m=0}^{\infty} R_{s,p}(m) \, \zeta^m \, x^{p+ms}.
	\end{equation}
	In particular, the exponents of $x$ appearing in the expansion of $w(z)^{-p}$
	are exactly $\{p + ms : m \ge 0\}$.
\end{lemma}

\begin{proof}
	From \eqref{eq:w-U}, $w^{-1} = (r/z)U = xU$, so $w^{-p} = x^p U^p$. Expanding
	$U^p$ via \eqref{eq:raney-generating} gives \eqref{eq:w-inverse-p}.
\end{proof}

\begin{remark}[Coefficient asymptotics]
	\label{rem:coeff-asymp}
	The square-root singularity \eqref{eq:U-sqrt} implies, via standard transfer
	theorems \cite{FlajoletSedgewick2009}, the universal large-$n$ asymptotics
	\begin{equation}\label{eq:Raney-asymp}
		R_{s,p}(n) = A_{s,p} \, n^{-3/2} \, \zeta_c^{-n} \bigl(1 + O(n^{-1})\bigr),
		\qquad n \to \infty,
	\end{equation}
	with an explicit amplitude $A_{s,p} > 0$ depending on $s$ and $p$.
	The exponent $-3/2$ is universal for simple algebraic singularities of
	square-root type. This is the only asymptotic information about the inverse map
	that enters the Hessian estimates in Section~\ref{sec:regime-I}. The detailed
	derivation is given in Appendix~\ref{app:raney-asymp}.
\end{remark}



\section{The Hessian of the \texorpdfstring{$\tau$}{tau}-function}
\label{sec:Hessian}


We now introduce the mixed Hessian of $\log\tau$ with respect to the
dispersionless Toda times. Throughout, we use the scale-free variables from
Section~\ref{sec:inversion}, which remove the trivial dependence on the
conformal radius and leave a Hessian depending only on the shape parameter
$\zeta=a/r$. We first recall the Wiegmann--Zabrodin kernel formula and define
the scale-invariant Hessian
(Subsection~\ref{subsec:kernel-hessian}). We then expand the kernel using the
Raney machinery from Section~\ref{sec:inversion}, obtaining an exact Gram
representation as a sum of rank-one operators
(Subsection~\ref{subsec:gram-representation}). Finally, we analyze the Gram
weights and locate the analytic threshold $\zeta_c$ at which the naive
unweighted Gram realization breaks down (Subsection~\ref{subsec:gram-weights}).

\subsection{Kernel representation and scale-invariant formulation}
\label{subsec:kernel-hessian}

The dispersionless Toda hierarchy associates to a simply connected domain
$D \subset \mathbb{C}$ a $\tau$-function whose logarithm $\mathcal{F}=\log\tau$
serves as a generating function for the harmonic moments. In the present
paper, however, we use this framework only for the polynomial family
\eqref{eq:map}.

For the maps \eqref{eq:map} with $0<\zeta<\zeta_{\mathrm{univ}}$, the boundary is
real-analytic and the Wiegmann--Zabrodin kernel identity applies. The mixed
second derivatives of $\mathcal{F}$ with respect to the Toda times $(t_m,
	\bar{t}_n)$ are then expressed via the inverse conformal map $w(z)$ as
follows~\cite{WiegmannZabrodin2000,BOOKtau}:
\begin{equation}\label{eq:hessian-original}
	\frac{\partial^2 \mathcal{F}}{\partial t_m \, \partial \bar{t}_n}
	= -mn \, [z^{-m}][\bar{z}'^{-n}] \log\!\left(1 - \frac{1}{w(z)\,\overline{w(z')}}\right),
\end{equation}
where $z$ and $z'$ are treated as independent complex variables and the
coefficient extraction $[z^{-m}]$ refers to the Laurent expansion at infinity.

The right-hand side of \eqref{eq:hessian-original} depends on the conformal
radius $r$ through the normalization $w(z) \sim z/r$ at infinity. This
dependence is inessential for spectral questions: it can be absorbed by
passing to the scale-free variable $x = r/z$ introduced in
Section~\ref{sec:inversion}. In terms of $x$, we have
$w(z)^{-1} = x \, U(x;\zeta)$ by \eqref{eq:w-U}, and the kernel becomes
\begin{equation}\label{eq:kernel-scaled}
	{K}(x, \bar{x}')
	:= -\log\!\left(1 - x\bar{x}' \, U(x;\zeta) \, \overline{U(x';\zeta)}\right).
\end{equation}
This expression depends on $(x, \bar{x}', \zeta)$ alone, with no residual
$r$-dependence.

\begin{definition}[Scale-invariant Hessian]
	\label{def:Hessian}
	The \emph{scale-invariant Hessian} is the infinite matrix $H = (H_{mn})_{m,n \ge 1}$
	defined by
	\begin{equation}\label{eq:H-def}
		H_{mn} := mn \, [x^{m}][\bar{x}'^{n}] \, {K}(x, \bar{x}').
	\end{equation}
\end{definition}

The prefactor $mn$ arises from the coefficient extraction in
\eqref{eq:hessian-original} and ensures that $H$ has the correct homogeneity
to act on mode sequences. Concretely, $H$ is the natural scale-invariant version of the Toda Hessian.
\begin{remark}[Operator viewpoint and role of renormalization]
	\label{rem:operator}
	The scale-invariant Hessian $H=(H_{mn})_{m,n\ge1}$ is naturally a positive
	semidefinite quadratic form on finitely supported sequences. Spectral
	statements require a compatible Hilbert realization, and the standard
	$\ell^2(\NN)$ topology is not stable at analytic criticality. The appropriate
	framework is the weighted Gram realization introduced later in
	Definition~\ref{def:weighted-conj}. All spectral statements in
	Sections~\ref{sec:regime-I}--\ref{sec:regime-2} refer to these renormalized
	operators.
\end{remark}

\subsection{Expansion of the kernel and block structure}
\label{subsec:kernel-expansion}

We now expand the kernel \eqref{eq:kernel-scaled} using the Raney formalism
from Section~\ref{sec:inversion}. The logarithm expands as
\begin{equation}\label{eq:log-series}
	{K}(x, \bar{x}')
	= \sum_{p=1}^{\infty} \frac{1}{p} \left( x\bar{x}' \, U(x;\zeta) \, \overline{U(x';\zeta)} \right)^p
	= \sum_{p=1}^{\infty} \frac{1}{p} \, x^p \bar{x}'^p \, U(x;\zeta)^p \, \overline{U(x';\zeta)}{}^p.
\end{equation}
By Lemma~\ref{lem:inverse-powers}, each power $U(x;\zeta)^p$ has the expansion
\[
	U(x;\zeta)^p = \sum_{m=0}^{\infty} R_{s,p}(m) \, (\zeta x^s)^m
	= \sum_{m=0}^{\infty} R_{s,p}(m) \, \zeta^m \, x^{sm},
\]
and similarly for the conjugate factor. Substituting into \eqref{eq:log-series}
and collecting powers of $x^{m} \bar{x}'^{n}$ yields:

\begin{proposition}[Kernel coefficients]
	\label{prop:kernel-coeff}
	The coefficient $[x^{m}][\bar{x}'^{n}] {K}$ vanishes unless
	$m\equiv n \pmod{s}$. When this congruence holds,
	\begin{equation}\label{eq:kernel-coeff}
		[x^{m}][\bar{x}'^{n}] {K}
		= \sum_{\substack{p \ge 1 \\ p \equiv m \pmod{s}}}
		\frac{1}{p} \, R_{s,p}(k) \, R_{s,p}(l) \, \zeta^{k+l}
	\end{equation}
	where $k = (m - p)/s$ and $l = (n - p)/s$.
\end{proposition}

\begin{proof}
	The monomial $x^{m} \bar{x}'^{n}$ arises from the $p$-th term in
	\eqref{eq:log-series} when $m = p + sk$ and $n = p + sl$ for some
	$k, l \ge 0$. This requires $m \equiv p \equiv n \pmod{s}$.
	Summing over all contributing values of $p$ gives \eqref{eq:kernel-coeff}.
\end{proof}

The selection rule has an immediate structural consequence for the Hessian.

\begin{corollary}[Block decomposition]
	\label{cor:block-decomp}
	The scale-invariant Hessian decomposes as
	\begin{equation}\label{eq:block-decomp}
		H = \bigoplus_{q=1}^{s} H^{(q)},
	\end{equation}
	where $H^{(q)}$ is the restriction of $H$ to indices $m, n \equiv q \pmod{s}$.
	Explicitly, identifying the $q$-th block with an operator on $\ell^2(\mathbb{N}_0)$
	via the correspondence $m = q + js \leftrightarrow j$, the matrix elements are
	\begin{equation}\label{eq:block-entry}
		H^{(q)}_{j_1 j_2} = H_{q + j_1 s, \, q + j_2 s}, \qquad j_1, j_2 \ge 0.
	\end{equation}
\end{corollary}

All spectral statements in this paper are formulated and proved blockwise. The
$s$-fold rotational symmetry of the conformal map thus reduces the spectral
problem to $s$ independent components.

\subsection{Gram representation}
\label{subsec:gram-representation}

The factorized structure of \eqref{eq:kernel-coeff}, namely, a product of Raney
coefficients, one depending on $m$ and one on $n$, suggests rewriting
the Hessian as a sum of rank-one operators. This \emph{Gram representation}
is the key algebraic structure underlying the spectral analysis.

\begin{proposition}[Gram representation]
	\label{prop:gram}
	Define vectors $\bm{v}^{(p)} = (v^{(p)}_{m})_{m \ge 1}$ by
	\begin{equation}\label{eq:gram-vector}
		v^{(p)}_{m} :=
		\begin{cases}
			\displaystyle \frac{m}{\sqrt{p}} \, R_{s,p}(k) \, \zeta^k,
			   & \text{if } m = p + ks \text{ for some } k \ge 0, \\[1ex]
			0, & \text{otherwise}.
		\end{cases}
	\end{equation}
	Then the scale-invariant Hessian admits the exact entrywise representation
	\begin{equation}\label{eq:gram-rep}
		H = \sum_{p=1}^{\infty} \bm{v}^{(p)} \otimes \bm{v}^{(p)*},
	\end{equation}
	where $\bm{v} \otimes \bm{v}^*$ denotes the rank-one operator
	$(\bm{v} \otimes \bm{v}^*)_{m n} = v_{m} \overline{v_{n}}$. Equivalently,
	\eqref{eq:gram-rep} holds as a quadratic-form identity on finitely supported
	sequences.
\end{proposition}

\begin{proof}
	Combining Definition~\ref{def:Hessian} with Proposition~\ref{prop:kernel-coeff},
	we have for $m = p + ks$ and $n = p + ls$:
	\[
		H_{mn} = mn \sum_{p' \equiv m} \frac{1}{p'} R_{s,p'}(k') R_{s,p'}(l') \zeta^{k'+l'},
	\]
	where $k' = (m - p')/s$ and $l' = (n - p')/s$. On the other hand,
	\[
		\sum_{p=1}^{\infty} v^{(p)}_{m} \overline{v^{(p)}_{n}}
		= \sum_{p \equiv m \equiv n} \frac{mn}{p} R_{s,p}(k) R_{s,p}(l) \zeta^{k+l},
	\]
	with $k = (m-p)/s$ and $l = (n-p)/s$. These expressions coincide.

	Finally, for each fixed $(m,n)$ only finitely many $p$ contribute
	(because $v^{(p)}_m=0$ unless $p\le m$ and $p\equiv m \!\!\pmod{s}$),
	so \eqref{eq:gram-rep} holds as an entrywise identity, equivalently as a
	quadratic-form identity on the vector space of finitely supported sequences on $\mathbb N$.
\end{proof}

\begin{remark}[Structure of the Gram representation]
	\label{rem:gram-structure}
	The Gram form \eqref{eq:gram-rep} makes positivity of $H$ manifest. Namely, it is a
	sum of positive semidefinite rank-one operators. More importantly, it
	separates two sources of complexity: (i) The \emph{individual} contribution of each logarithmic mode $p$,
	encoded in the vector $\bm{v}^{(p)}$, and (ii) the \emph{collective} mixing of modes within each symmetry sector,
	which determines the eigenvalue distribution.

	The spectral behavior of $H$ is controlled by the interplay between these
	two effects. As we show next, the critical phenomenon at $\zeta = \zeta_c$
	arises from the borderline summability of the Gram entries, rather than by a singularity in any individual term.
\end{remark}

\subsection{Gram weights and the analytic threshold}
\label{subsec:gram-weights}

The Gram representation expresses $H$ as a superposition of rank-one
contributions. Whether this sum defines a bounded operator on the naive
unweighted $\ell^2$-space depends on the size of the individual terms,
measured by their squared norms.

\begin{definition}[Gram weights]
	\label{def:gram-weights}
	For $p \ge 1$, the \emph{$p$-th Gram weight} is
	\begin{equation}\label{eq:gram-weight-def}
		\sigma_p(\zeta) := \|\bm{v}^{(p)}\|_{\ell^2}^2
		= \sum_{m=0}^{\infty} \frac{(p+ms)^2}{p} \, R_{s,p}(m)^2 \, \zeta^{2m}.
	\end{equation}
\end{definition}

The Gram weights are the natural quantities controlling operator bounds for
the Hessian. Their behavior as $\zeta$ varies reveals the analytic threshold.

\begin{proposition}[Threshold for Gram weights]
	\label{prop:gram-divergence}
	For each fixed $p \ge 1$, one has $\sigma_p(\zeta) < \infty$ if and only if
	$\zeta < \zeta_c$. In particular, the analytic threshold for the Gram weights is exactly~$\zeta_c$.
\end{proposition}

\begin{proof}
	By Proposition~\ref{prop:raney-asymp}, the Raney coefficients satisfy
	$R_{s,p}(m) \sim A_{s,p} \, m^{-3/2} \, \zeta_c^{-m}$ as
	$m \to \infty$. Hence the summand in \eqref{eq:gram-weight-def} has the
	borderline form
	\[
		\frac{(p+ms)^2}{p} R_{s,p}(m)^2 \zeta^{2m}
		\sim \frac{s^2 A_{s,p}^2}{p} \, m^{-1}
		\left(\frac{\zeta}{\zeta_c}\right)^{2m}.
	\]
	The factor $m^{-1}$ is summable exactly when $\zeta<\zeta_c$, and not
	summable at $\zeta=\zeta_c$. This proves the claim.
\end{proof}

\begin{remark}[Nature of the threshold]
	\label{rem:threshold-nature}
	Proposition~\ref{prop:gram-divergence} shows that $\zeta_c$ is exactly the
	threshold at which the vectors $\bm v^{(p)}$ cease to belong to $\ell^2$.
	Hence the elementary Gram realization of this subsection is valid only for
	$\zeta<\zeta_c$. The coefficient matrix \(H\), however, remains well defined
	throughout the geometric regime \(\zeta<\zeta_{\mathrm{univ}}\), and the
	weighted framework of Section~\ref{sec:regime-I} is introduced to study the
	singular behavior near \(\zeta_c\) on a fixed Hilbert scale.
\end{remark}

The strict inequality \(\zeta_c<\zeta_{\mathrm{univ}}\), proved later in
Proposition~\ref{prop:separation}, produces the intermediate regime
\(\zeta_c<\zeta<\zeta_{\mathrm{univ}}\) in which spectral criticality
precedes geometric breakdown. The detailed analysis of these two regimes is
carried out in Sections~\ref{sec:regime-I} and~\ref{sec:regime-2}.




\section{Subcritical phase \texorpdfstring{$(0<\zeta<\zeta_c)$}{(0<zeta<zeta_c)}}
\label{sec:regime-I}


We now turn to the operator-theoretic picture in the analytic subcritical
regime. Proposition~\ref{prop:gram} gives a canonical Gram factorization of
each symmetry block of the Hessian. Although the unrenormalized Gram weights
$\sigma_p(\zeta)$ diverge as $\zeta\uparrow\zeta_c$
(Proposition~\ref{prop:gram-divergence}), an explicit diagonal weighting
renormalizes the column modes. The resulting compact positive operators have a
single logarithmically diverging eigenvalue, while the remaining soft
spectrum stays bounded and admits a well-defined limiting description after
compression.

\subsection{Operator framework, block factorization, and weighted renormalization}
\label{subsec:gram-operators}

Fix a symmetry sector \(q\in\{1,\dots,s\}\), and write
\[
	p_j:=q+js,\qquad j\in\NN_0,
\]
for the indices belonging to that block; see
Corollary~\ref{cor:block-decomp}. Restricting the Gram vectors
\(\bm v^{(p)}\) from \eqref{eq:gram-vector} to the \(q\)-th block, we obtain
a synthesis map \(V(\zeta)\), initially defined on finitely supported
sequences, whose \(j\)-th column is the restriction of
\(\bm v^{(p_j)}\). The block Hessian and block Gram operators are then
\begin{equation}\label{eq:Hq-VV}
	H^{(q)}(\zeta)=V(\zeta)V(\zeta)^*,
	\qquad
	G^{(q)}(\zeta)=V(\zeta)^*V(\zeta).
\end{equation}
At this stage, \(V(\zeta)\) is used only on finitely supported sequences, and
\eqref{eq:Hq-VV} is an identity of matrix coefficients (equivalently, of
quadratic forms on finitely supported vectors).

The difficulty is that the unweighted realization is not stable as
\(\zeta\uparrow\zeta_c\). Indeed, Proposition~\ref{prop:raney-uniform}
shows that the Raney coefficients satisfy
\[
	R_{s,p}(m)\le C_s\,p\,M^{p}\,\zeta_c^{-m}\,m^{-3/2},
	\qquad
	M:=\frac{s}{s-1}.
\]
Accordingly, the \(j\)-th column of \(V(\zeta)\) carries the noncritical
background growth \(p_jM^{p_j}\) in the block index \(j\), together with the
borderline transfer tail \(m^{-3/2}\) in the summation variable \(m\). To
obtain a compact operator on a fixed Hilbert space, one therefore rescales
the columns by an explicit diagonal weight that removes this background
growth but leaves the genuine logarithmic singularity untouched.

\begin{definition}[Weighted realization]
	\label{def:weighted-conj}
	Fix \(\beta>0\), set \(M:=\frac{s}{s-1}\), and define
	\begin{equation}\label{eq:weights}
		w_j:=p_j^{\,\frac32+\beta}M^{p_j},
		\qquad
		j\in\NN_0.
	\end{equation}
	Let
	\begin{equation}\label{eq:weighted-Hilbert}
		\mathcal H_\beta
		:=
		\Bigl\{x=(x_j)_{j\ge0}:\,
		\|x\|_{\mathcal H_\beta}^2
		=
		\sum_{j\ge0}|x_j|^2w_j^2<\infty\Bigr\},
	\end{equation}
	with inner product
	\[
		\langle x,y\rangle_{\mathcal H_\beta}
		:=
		\sum_{j\ge0}x_j\overline{y_j}\,w_j^2.
	\]
	Let \(\mathcal W:\mathcal H_\beta\to\ell^2(\NN_0)\) be the unitary map:
	\[
		(\mathcal Wx)_j=w_jx_j.
	\]
	Identifying \(\mathcal W\) with the diagonal operator
	\(\operatorname{diag}(w_j)_{j\ge0}\) on \(\ell^2(\NN_0)\) via the canonical
	basis, the renormalized synthesis operator is
	\[
		\widetilde V(\zeta):=V(\zeta)\mathcal W^{-1}:\ell^2(\NN_0)\to\ell^2(\NN_0),
	\]
	initially defined on finitely supported sequences.
	Proposition~\ref{prop:subcrit-compact} shows that it extends uniquely by
	continuity to a Hilbert--Schmidt operator. Its associated weighted Gram and
	Hessian operators are
	\begin{equation}\label{eq:Gtilde}
		\widetilde G^{(q)}(\zeta)
		:=
		\widetilde V(\zeta)^*\widetilde V(\zeta),
		\qquad
		\widetilde H^{(q)}(\zeta)
		:=
		\widetilde V(\zeta)\widetilde V(\zeta)^*.
	\end{equation}
\end{definition}

The structure of the weights in \eqref{eq:weights} is dictated by the Raney
asymptotics. The factor \(M^{p_j}\) compensates the exponential dependence on
the block index. The exponent \(3/2\) matches the borderline tail
\(m^{-3/2}\), and the additional parameter \(\beta>0\) provides a small
margin that turns the residual \(p_j^{-2}\) behavior into an absolutely
summable sequence. In particular, the weighted realization is
\(\zeta\)-independent.

\begin{remark}[Intrinsic object versus weighted realization]
	\label{rem:renorm-intrinsic}
	The coefficient matrix of the Hessian block is the intrinsic object. The weights \eqref{eq:weights} neither alter these coefficients nor remove the critical instability. Their role is to place the same coefficients in a fixed $\zeta$-independent Hilbert scale in which the singular part is realized as a compact-plus-rank-one operator.

	The singular direction is already present at the coefficient level and is
	encoded by the intrinsic amplitudes \(d_j^{(q)}\); see
	\eqref{eq:app-d-def}. The vector
	\[
		\widetilde{\bm d}^{(q)}
		=
		\bigl(d_j^{(q)}/w_j\bigr)_{j\ge0}\in\ell^2(\NN_0)
	\]
	is simply its realization in the weighted scale. Varying \(\beta>0\) changes
	this realization, but not the underlying one-dimensional singular direction.
	Thus Theorems~\ref{thm:rankone-decomp}
	and~\ref{thm:spectral-asymptotics} are statements about the weighted
	realization for fixed \(\beta>0\), whereas their intrinsic content is the
	existence of a single dominant singular direction in each symmetry sector.
\end{remark}

The first consequence of this construction is compactness for every fixed
subcritical parameter.

\begin{proposition}[Subcritical compactness]
	\label{prop:subcrit-compact}
	Fix \(q\in\{1,\dots,s\}\) and \(\beta>0\). For every \(0<\zeta<\zeta_c\),
	the operator
	\[
		\widetilde V(\zeta):\ell^2(\NN_0)\to\ell^2(\NN_0)
	\]
	is Hilbert--Schmidt. Consequently,
	\(\widetilde G^{(q)}(\zeta)\) and \(\widetilde H^{(q)}(\zeta)\) are
	trace class, hence compact.
\end{proposition}

\begin{proof}
	By construction, the \(j\)-th column of \(\widetilde V(\zeta)\) is
	\[
		\widetilde V(\zeta)e_j
		=
		\frac{\bm v^{(p_j)}}{w_j},
	\]
	restricted to the \(q\)-block. Therefore
	\begin{equation}\label{eq:HS-sigma-sum}
		\|\widetilde V(\zeta)\|_{\mathrm{HS}}^2
		=
		\sum_{j\ge0}\|\widetilde V(\zeta)e_j\|_{\ell^2}^2
		=
		\sum_{j\ge0}\frac{\|\bm v^{(p_j)}\|_{\ell^2}^2}{w_j^2}
		=
		\sum_{j\ge0}\frac{\sigma_{p_j}(\zeta)}{w_j^2}.
	\end{equation}
	It is thus enough to bound \(\sigma_p(\zeta)\) uniformly in \(p\). Using the closed form for \(R_{s,p}(m)\) from
	Proposition~\ref{prop:raney-lagrange}, the Gram weights may be written as
	\begin{equation}\label{eq:sigma-binomial-form}
		\sigma_p(\zeta)
		=
		p\sum_{m\ge0}\binom{sm+p}{m}^2\zeta^{2m},
		\qquad p\ge1.
	\end{equation}
	We claim that there exists \(C_s>0\) such that, for all \(p\ge1\) and
	\(m\ge1\),
	\begin{equation}\label{eq:binom-uniform-bound}
		\binom{sm+p}{m}
		\le
		C_s\,M^p\,\zeta_c^{-m}\,m^{-1/2}.
	\end{equation}
	Indeed, set $
		N:=sm+p$, and $K:=N-m=(s-1)m+p$.
	A two-sided Stirling bound gives
	\[
		\binom{N}{m}
		\le
		C\,\frac{N^{N+\frac12}}{m^{m+\frac12}K^{K+\frac12}}
		=
		C\,\frac1{\sqrt m}\sqrt{\frac NK}
		\exp\bigl(N\log N-m\log m-K\log K\bigr),
	\]
	with an absolute constant \(C>0\). Writing \(t:=p/m\ge0\), so that
	\(N=m(s+t)\) and \(K=m(s-1+t)\), the exponent becomes $
		m((s+t)\log(s+t)-(s-1+t)\log(s-1+t))$.
	By the convexity estimate used in
	Proposition~\ref{prop:raney-uniform},
	\[
		(s+t)\log(s+t)-(s-1+t)\log(s-1+t)
		\le
		\log(\zeta_c^{-1})+t\log M,
	\]
	hence
	\[
		\exp\bigl(N\log N-m\log m-K\log K\bigr)
		\le
		\zeta_c^{-m}M^p.
	\]
	Moreover,
	\[
		\sqrt{\frac NK}
		=
		\sqrt{\frac{s+t}{s-1+t}}
		\le
		\sqrt{\frac{s}{s-1}},
	\]
	uniformly in \(t\ge0\). This proves \eqref{eq:binom-uniform-bound}.

	Now fix \(0<\zeta<\zeta_c\) and set \(\eta:=\zeta/\zeta_c\in(0,1)\).
	Combining \eqref{eq:sigma-binomial-form} with
	\eqref{eq:binom-uniform-bound}, and separating the term \(m=0\), we obtain
	\[
		\sigma_p(\zeta)
		=
		p+p\sum_{m\ge1}\binom{sm+p}{m}^2\zeta^{2m}
		\le
		p+p\,C_s^2M^{2p}\sum_{m\ge1}m^{-1}\eta^{2m}.
	\]
	Since \(\sum_{m\ge1}m^{-1}\eta^{2m}=-\log(1-\eta^2)<\infty\), this yields
	\begin{equation}\label{eq:sigma-subcritical-bound}
		\sigma_p(\zeta)\le C_s'(\zeta)\,p\,M^{2p},
	\end{equation}
	for some finite constant \(C_s'(\zeta)\) depending on \(\zeta\), but not on
	\(p\). Finally, using \(w_j^2=p_j^{3+2\beta}M^{2p_j}\) and
	\eqref{eq:HS-sigma-sum},
	\[
		\frac{\sigma_{p_j}(\zeta)}{w_j^2}
		\le
		C_s'(\zeta)\,\frac{p_jM^{2p_j}}{p_j^{3+2\beta}M^{2p_j}}
		=
		C_s'(\zeta)\,p_j^{-2-2\beta}.
	\]
	Because \(p_j=q+js\sim sj\) and \(\beta>0\), the series
	\(\sum_{j\ge0}p_j^{-2-2\beta}\) converges. Hence
	\(\widetilde V(\zeta)\) is Hilbert--Schmidt. Therefore
	\[
		\widetilde G^{(q)}(\zeta)=\widetilde V(\zeta)^*\widetilde V(\zeta),
		\qquad
		\widetilde H^{(q)}(\zeta)=\widetilde V(\zeta)\widetilde V(\zeta)^*,
	\]
	are trace class, and in particular compact.
\end{proof}

\begin{lemma}[Isospectrality of the weighted block realizations]
	\label{lem:block-isospectral}
	For each \(0<\zeta<\zeta_c\), the compact positive operators
	\(\widetilde H^{(q)}(\zeta)\) and \(\widetilde G^{(q)}(\zeta)\) have the
	same nonzero eigenvalues, counted with multiplicities.
\end{lemma}

\begin{proof}
	By Proposition~\ref{prop:subcrit-compact}, the operator
	\(\widetilde V(\zeta):\ell^2(\NN_0)\to\ell^2(\NN_0)\) is bounded.
	Therefore the standard relation between \(\widetilde V\widetilde V^*\) and
	\(\widetilde V^*\widetilde V\) applies. If
	\(\widetilde G^{(q)}(\zeta)x=\mu\,x\) with \(\mu>0\), then
	\(\widetilde V(\zeta)x\neq0\) and
	\[
		\widetilde H^{(q)}(\zeta)\bigl(\widetilde V(\zeta)x\bigr)
		=
		\widetilde V(\zeta)\widetilde V(\zeta)^*\widetilde V(\zeta)x
		=
		\widetilde V(\zeta)\widetilde G^{(q)}(\zeta)x
		=
		\mu\,\widetilde V(\zeta)x.
	\]
	The converse implication is identical, with \(\widetilde V(\zeta)^*\) in
	place of \(\widetilde V(\zeta)\). Multiplicities are preserved by the same
	correspondence.
\end{proof}

\begin{corollary}[Transfer from the weighted Gram block to the weighted Hessian block]
	\label{cor:block-transfer}
	Fix \(q\in\{1,\dots,s\}\), \(\beta>0\), and \(0<\zeta<\zeta_c\). Then
	\(\widetilde G^{(q)}(\zeta)\) and \(\widetilde H^{(q)}(\zeta)\) have the
	same nonzero eigenvalues, counted with multiplicities. In particular, every
	nonzero-eigenvalue statement proved below for \(\widetilde G^{(q)}(\zeta)\)
	holds verbatim for \(\widetilde H^{(q)}(\zeta)\).
\end{corollary}

\begin{proof}
	The first statement is exactly Lemma~\ref{lem:block-isospectral}. The
	remaining assertion follows because all spectral conclusions in
	Sections~\ref{subsec:rankone-spike}--\ref{subsec:spectral-asymptotics} are
	formulated only in terms of the nonzero eigenvalues and their
	multiplicities.
\end{proof}

\subsection{Rank-one logarithmic spike at \texorpdfstring{$\zeta_c$}{zeta_c}}
\label{subsec:rankone-spike}

We now isolate the singular part of the weighted block Gram operator as
\(\zeta\uparrow\zeta_c\). The point is that, after the weighted realization of
Section~\ref{subsec:gram-operators}, the borderline divergence is carried by a
single vector \(\widetilde{\bm d}^{(q)}\), whereas the remaining part stays
compact and converges in norm. This is the operator-theoretic core of
Theorem~\ref{thm:A}.

\begin{theorem}[Rank-one decomposition]
	\label{thm:rankone-decomp}
	Fix \(q\in\{1,\dots,s\}\) and \(\beta>0\). There exist a nonzero vector
	\(\widetilde{\bm d}^{(q)}\in\ell^2(\NN_0)\), independent of \(\zeta\), and a
	family of compact operators \(\widetilde C^{(q)}(\zeta)\) on \(\ell^2(\NN_0)\)
	such that, for every \(0<\zeta<\zeta_c\),
	\begin{equation}\label{eq:rankone-decomp}
		\widetilde G^{(q)}(\zeta)
		=
		L(\zeta)\,\widetilde{\bm d}^{(q)}\otimes \widetilde{\bm d}^{(q)*}
		+
		\widetilde C^{(q)}(\zeta),
		\qquad
		L(\zeta):=\log\frac{1}{1-\zeta^2/\zeta_c^2}.
	\end{equation}
	Moreover,
	\[
		\sup_{0<\zeta<\zeta_c}\|\widetilde C^{(q)}(\zeta)\|<\infty,
		\qquad
		\widetilde C^{(q)}(\zeta)\to \widetilde C_*^{(q)}
		\quad\text{in operator norm as }\zeta\uparrow\zeta_c,
	\]
	for some compact limit operator \(\widetilde C_*^{(q)}\).

	The vector \(\widetilde{\bm d}^{(q)}\) is described explicitly in
	Appendix~\ref{app:rankone-form}. Its intrinsic, weight-independent amplitudes
	\(d^{(q)}\) are defined in \eqref{eq:app-d-def}, and
	\[
		\widetilde d^{(q)}_j=\frac{d^{(q)}_j}{w_j}
	\]
	is their realization in the weighted scale.
\end{theorem}

\begin{proof}
	The detailed coefficient analysis is contained in
	Appendix~\ref{app:rankone-form}. Here we only present the reduction that
	produces \eqref{eq:rankone-decomp}.

	\smallskip
	\noindent\emph{Entrywise extraction of the logarithmic term.}
	By Lemma~\ref{lem:Gram-entry}, each matrix element of
	\(\widetilde G^{(q)}(\zeta)\) is given by a single series in the summation
	index \(m\). Lemma~\ref{lem:entrywise-log} applies the uniform Raney
	expansion of Lemma~\ref{lem:raney-uniform-expansion} to the tail of this
	series and controls the initial range by the global bound of
	Proposition~\ref{prop:raney-uniform}. As a result, if
	\[
		\eta:=\frac{\zeta}{\zeta_c}\in(0,1),
	\]
	and
	\[
		(K_\eta)_{j_1j_2}
		:=
		\eta^{|j_1-j_2|}\,
		\widetilde d^{(q)}_{j_1}\widetilde d^{(q)}_{j_2},
		\qquad
		j_1,j_2\in\NN_0,
	\]
	then
	\begin{equation}\label{eq:rankone-preliminary}
		\widetilde G^{(q)}(\zeta)
		=
		L(\zeta)\,K_\eta
		+
		R^{(q)}(\zeta),
	\end{equation}
	where \(R^{(q)}(\zeta)\) is Hilbert--Schmidt,
	\[
		\sup_{0<\zeta<\zeta_c}\|R^{(q)}(\zeta)\|_{\mathrm{HS}}<\infty,
	\]
	and
	\[
		R^{(q)}(\zeta)\to \widetilde C_*^{(q)}
		\qquad\text{in Hilbert--Schmidt norm as }\zeta\uparrow\zeta_c.
	\]
	In particular, \(R^{(q)}(\zeta)\) is compact for each \(\zeta\), and the
	convergence also holds in operator norm. We denote this operator-norm
	limit by \(\widetilde C_*^{(q)}\). Thus the \(R_*^{(q)}\) of
	Appendix~\ref{app:rankone-form} and the \(\widetilde C_*^{(q)}\) used in
	the main text are the same operator.

	\smallskip
	\noindent\emph{Removal of the residual Toeplitz factor.}
	Set
	\[
		K_1:=\widetilde{\bm d}^{(q)}\otimes \widetilde{\bm d}^{(q)*}.
	\]
	The kernel \(K_\eta\) differs from \(K_1\) only by the Toeplitz damping factor
	\(\eta^{|j_1-j_2|}\). By Lemma~\ref{lem:toeplitz-removal},
	\begin{equation}\label{eq:toeplitz-removal-proof}
		L(\zeta)\,\|K_\eta-K_1\|\longrightarrow0,
		\qquad \zeta\uparrow\zeta_c.
	\end{equation}
	Therefore
	\[
		L(\zeta)\,K_\eta
		=
		L(\zeta)\,K_1
		+
		L(\zeta)\,(K_\eta-K_1),
	\]
	and the second term is negligible at the operator level.

	\smallskip
	\noindent\emph{Final decomposition.}
	Combining \eqref{eq:rankone-preliminary} with
	\eqref{eq:toeplitz-removal-proof}, we obtain
	\[
		\widetilde G^{(q)}(\zeta)
		=
		L(\zeta)\,K_1
		+
		\Bigl(R^{(q)}(\zeta)+L(\zeta)(K_\eta-K_1)\Bigr).
	\]
	Define
	\[
		\widetilde C^{(q)}(\zeta)
		:=
		R^{(q)}(\zeta)+L(\zeta)(K_\eta-K_1).
	\]
	Each \(\widetilde C^{(q)}(\zeta)\) is compact, since
	\(R^{(q)}(\zeta)\) is Hilbert--Schmidt and \(K_\eta-K_1\) is
	Hilbert--Schmidt by Lemma~\ref{lem:toeplitz-removal}.
	Moreover,
	\[
		\sup_{0<\zeta<\zeta_c}\|\widetilde C^{(q)}(\zeta)\|
		\le
		\sup_{0<\zeta<\zeta_c}\|R^{(q)}(\zeta)\|
		+
		\sup_{0<\zeta<\zeta_c}L(\zeta)\|K_\eta-K_1\|
		<\infty,
	\]
	and
	\[
		\widetilde C^{(q)}(\zeta)\to \widetilde C_*^{(q)}
		\qquad\text{in operator norm as }\zeta\uparrow\zeta_c,
	\]
	because \(R^{(q)}(\zeta)\to \widetilde C_*^{(q)}\) in Hilbert--Schmidt norm
	and \(L(\zeta)(K_\eta-K_1)\to0\) in operator norm. Since \(K_1=\widetilde{\bm d}^{(q)}\otimes \widetilde{\bm d}^{(q)*}\), this
	is exactly \eqref{eq:rankone-decomp}.
\end{proof}


\subsection{Spectral asymptotics}
\label{subsec:spectral-asymptotics}

We now convert the rank-one decomposition of
Theorem~\ref{thm:rankone-decomp} into spectral statement. The conclusion is
that the singular part produces one stiff eigenvalue in each symmetry
block, while the remaining spectrum stays bounded and converges to the
compressed limit operator.

\begin{theorem}[Spectral asymptotics at the threshold]
	\label{thm:spectral-asymptotics}
	Fix \(q\in\{1,\dots,s\}\) and \(\beta>0\). Let
	\[
		\mu_1^{(q)}(\zeta)\ge \mu_2^{(q)}(\zeta)\ge\cdots\ge0
	\]
	denote the eigenvalues of \(\widetilde G^{(q)}(\zeta)\), counted with
	multiplicity, and let \(\widetilde{\bm d}^{(q)}\neq0\) be the spike vector
	from Theorem~\ref{thm:rankone-decomp}. Set
	\[
		\Gamma^{(q)}:=\|\widetilde{\bm d}^{(q)}\|_{\ell^2}^2.
	\]
	Then, as \(\zeta\uparrow\zeta_c\),
	\begin{equation}\label{eq:mu-sing}
		\mu_1^{(q)}(\zeta)=\Gamma^{(q)}\,L(\zeta)+O(1),
		\qquad
		\sup_{0<\zeta<\zeta_c}\mu_k^{(q)}(\zeta)<\infty
		\quad (k\ge2).
	\end{equation}
	In particular, in the \(q\)-th symmetry block exactly one eigenvalue
	diverges logarithmically, and this eigenvalue is simple for \(\zeta\)
	sufficiently close to \(\zeta_c\).

	Moreover, if \(\psi_1(\zeta)\) is a unit eigenvector associated with
	\(\mu_1^{(q)}(\zeta)\), then after fixing its phase one has
	\begin{equation}\label{eq:psi-align}
		\psi_1(\zeta)\longrightarrow
		\frac{\widetilde{\bm d}^{(q)}}{\|\widetilde{\bm d}^{(q)}\|_{\ell^2}}
		\qquad\text{in }\ell^2
		\quad\text{as }\zeta\uparrow\zeta_c.
	\end{equation}
	Finally,
	\begin{equation}\label{eq:mu1-asymp}
		\mu_1^{(q)}(\zeta)
		=
		\Gamma^{(q)}\,L(\zeta)
		+
		\Lambda_{\mathrm{fin}}^{(q)}
		+o(1),
		\qquad \zeta\uparrow\zeta_c,
	\end{equation}
	where
	\[
		\Lambda_{\mathrm{fin}}^{(q)}
		:=
		\frac{\langle \widetilde{\bm d}^{(q)},
			\widetilde C_*^{(q)}\widetilde{\bm d}^{(q)}\rangle}
		{\|\widetilde{\bm d}^{(q)}\|_{\ell^2}^2}.
	\]
\end{theorem}

Theorem~\ref{thm:spectral-asymptotics} has a clear numerical signature:
after weighted renormalization, exactly one eigenvalue in each symmetry block
grows on the logarithmic scale \(L(\zeta)\), while the remaining eigenvalues
stay bounded. In all figures below, \(N\) denotes the truncation size of the
finite weighted block matrix used for numerical diagonalization. Figure~\ref{fig:eigenvalue-trajectories} illustrates this
separation for the block operator \(\widetilde G^{(q)}(\zeta)\) in the sector
\(q=1\) for \(s=3\) and \(s=5\). The top row shows that the leading eigenvalue
\(\mu_1^{(q)}(\zeta)\) is asymptotically affine in \(L(\zeta)\), whereas the
bottom row confirms that after division by \(L(\zeta)\) only the stiff eigenvalue
survives and all soft eigenvalues tend to zero.

\begin{figure}[t]
	\centering
	\includegraphics[width=\textwidth]{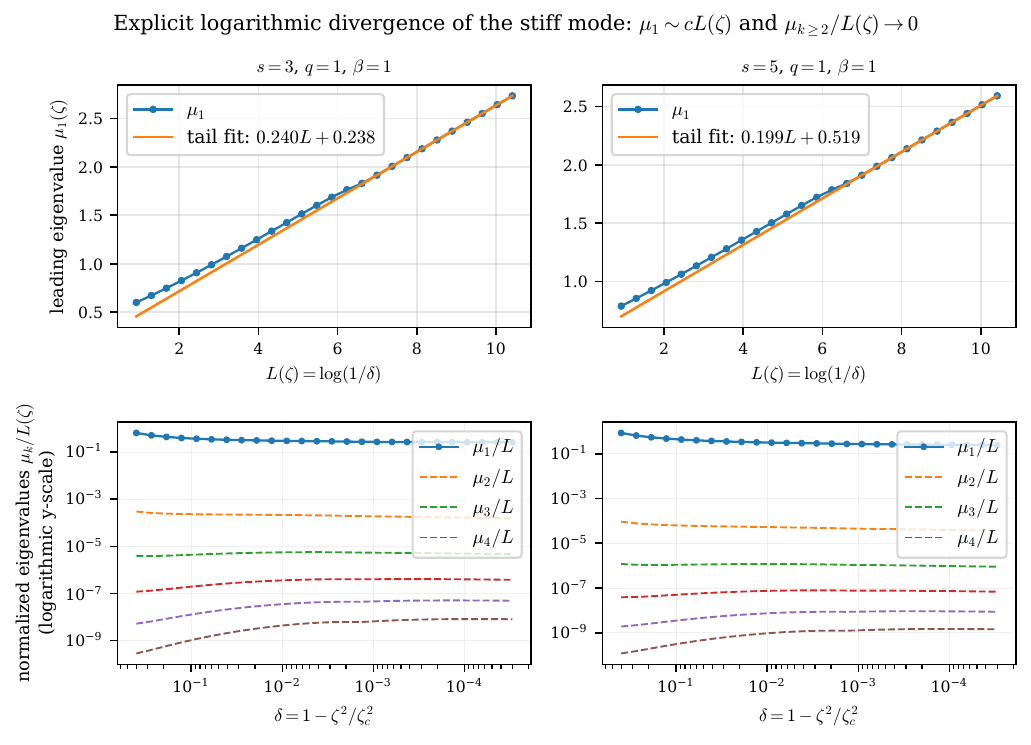}
	\caption{
		Logarithmic spectral asymptotics of the weighted Gram block
		\(\widetilde G^{(q)}(\zeta)\) for \(s=3,5\) in sector \(q=1\)
		(\(\beta=1\), \(N=30\)).
		Top row: the leading eigenvalue \(\mu_1^{(q)}(\zeta)\) plotted against $L(\zeta)$,
		showing the asymptotically affine law
		\(\mu_1^{(q)}(\zeta)=\Gamma^{(q)}\,L(\zeta)+O(1)\). The dashed line is a tail
		linear fit.
		Bottom row: the normalized eigenvalues
		\(\mu_k^{(q)}(\zeta)/L(\zeta)\). The ratio
		\(\mu_1^{(q)}(\zeta)/L(\zeta)\) approaches a positive constant, whereas
		\(\mu_k^{(q)}(\zeta)/L(\zeta)\to0\) for \(k\ge2\).
	}
	\label{fig:eigenvalue-trajectories}
\end{figure}

\begin{proof}[Proof of Theorem~\ref{thm:spectral-asymptotics}]
	Write $
		\widetilde{\bm d}:=\widetilde{\bm d}^{(q)},
		\
		\widehat{\bm d}:=\widetilde{\bm d}/\|\widetilde{\bm d}\|_{\ell^2},
		\
		P:=\widehat{\bm d}\otimes\widehat{\bm d}^{*}$. By Theorem~\ref{thm:rankone-decomp},
	\begin{equation}\label{eq:G-rankone-tight}
		\widetilde G^{(q)}(\zeta)
		=
		\Gamma^{(q)}\,L(\zeta)\,P+\widetilde C^{(q)}(\zeta),
	\end{equation}
	where each \(\widetilde C^{(q)}(\zeta)\) is compact self-adjoint,
	\[
		K_q:=\sup_{0<\zeta<\zeta_c}\|\widetilde C^{(q)}(\zeta)\|<\infty,
	\]
	and $\widetilde C^{(q)}(\zeta)\to \widetilde C_*^{(q)}$
		in operator norm as $\zeta\uparrow\zeta_c$.

	\smallskip
	\noindent\emph{Only one eigenvalue can diverge.}
	If \(x\perp\widetilde{\bm d}\) and \(\|x\|=1\), then \(Px=0\), hence
	\[
		\langle x,\widetilde G^{(q)}(\zeta)x\rangle
		=
		\langle x,\widetilde C^{(q)}(\zeta)x\rangle.
	\]
	Using the one-dimensional trial space \(\Span\{\widetilde{\bm d}\}\) in the
	min--max characterization of \(\mu_2^{(q)}(\zeta)\), we obtain
	\begin{equation}\label{eq:mu2-bound-tight}
		\mu_2^{(q)}(\zeta)
		\le
		\sup_{\substack{\|x\|=1\\x\perp\widetilde{\bm d}}}
		\langle x,\widetilde G^{(q)}(\zeta)x\rangle
		\le
		\|\widetilde C^{(q)}(\zeta)\|
		\le K_q.
	\end{equation}
	Therefore every eigenvalue with index \(k\ge2\) remains uniformly bounded.

	On the other hand, the Rayleigh quotient of \(\widehat{\bm d}\) gives
	\[
		\mu_1^{(q)}(\zeta)
		\ge
		\langle \widehat{\bm d},\widetilde G^{(q)}(\zeta)\widehat{\bm d}\rangle
		=
		\Gamma^{(q)}\,L(\zeta)
		+
		\langle \widehat{\bm d},
		\widetilde C^{(q)}(\zeta)\widehat{\bm d}\rangle
		\ge
		\Gamma^{(q)}\,L(\zeta)-K_q.
	\]
	Since \(L(\zeta)\to+\infty\), it follows that \(\mu_1^{(q)}(\zeta)\to+\infty\)
	while \(\mu_k^{(q)}(\zeta)\), \(k\ge2\), stay bounded. This proves
	\eqref{eq:mu-sing}, and the simplicity of \(\mu_1^{(q)}(\zeta)\) for
	\(\zeta\) close to \(\zeta_c\) follows from the spectral gap
	\[
		\mu_1^{(q)}(\zeta)-\mu_2^{(q)}(\zeta)\to+\infty.
	\]

	\smallskip
	\noindent\emph{The top eigenvector aligns with the spike direction.}
	Let \(\psi_1(\zeta)\) be a unit eigenvector corresponding to
	\(\mu_1^{(q)}(\zeta)\). Applying \(I-P\) to the eigenvalue equation
	\[
		\widetilde G^{(q)}(\zeta)\psi_1(\zeta)
		=
		\mu_1^{(q)}(\zeta)\psi_1(\zeta)
	\]
	and using \((I-P)P=0\), we obtain
	\[
		(I-P)\widetilde C^{(q)}(\zeta)\psi_1(\zeta)
		=
		\mu_1^{(q)}(\zeta)(I-P)\psi_1(\zeta).
	\]
	Hence
	\begin{equation}\label{eq:psi-align-tight}
		\|(I-P)\psi_1(\zeta)\|
		\le
		\frac{\|\widetilde C^{(q)}(\zeta)\|}{\mu_1^{(q)}(\zeta)}
		\le
		\frac{K_q}{\mu_1^{(q)}(\zeta)}
		=
		O\bigl(L(\zeta)^{-1}\bigr).
	\end{equation}
	If $a(\zeta):=\langle \psi_1(\zeta),\widehat{\bm d}\rangle$, then
	\[
		|a(\zeta)|^2
		=
		1-\|(I-P)\psi_1(\zeta)\|^2
		=
		1+O\bigl(L(\zeta)^{-2}\bigr).
	\]
	After fixing the phase so that \(a(\zeta)\ge0\), we obtain
	\eqref{eq:psi-align}.

	\smallskip
	\noindent\emph{Asymptotics of the diverging eigenvalue.}
	Since \(\psi_1(\zeta)\) is normalized,
	\[
		\mu_1^{(q)}(\zeta)
		=
		\langle \psi_1(\zeta),\widetilde G^{(q)}(\zeta)\psi_1(\zeta)\rangle
		=
		\Gamma^{(q)}\,L(\zeta)\,|a(\zeta)|^2
		+
		\langle \psi_1(\zeta),
		\widetilde C^{(q)}(\zeta)\psi_1(\zeta)\rangle.
	\]
	By \eqref{eq:psi-align-tight}, $
		\Gamma^{(q)}\,L(\zeta)\,|a(\zeta)|^2
		=
		\Gamma^{(q)}\,L(\zeta)+o(1)$,
	while \eqref{eq:psi-align} and the operator-norm convergence of
	\(\widetilde C^{(q)}(\zeta)\) imply
	\[
		\langle \psi_1(\zeta),
		\widetilde C^{(q)}(\zeta)\psi_1(\zeta)\rangle
		\to
		\langle \widehat{\bm d},\widetilde C_*^{(q)}\widehat{\bm d}\rangle
		=
		\frac{\langle \widetilde{\bm d},
			\widetilde C_*^{(q)}\widetilde{\bm d}\rangle}
		{\|\widetilde{\bm d}\|_{\ell^2}^2}.
	\]
	This is exactly \eqref{eq:mu1-asymp}.
\end{proof}

\FloatBarrier

The remaining spectral behaviour is governed by the compression of the limiting
remainder to the orthogonal complement of the spike direction.

\begin{proposition}[Convergence of the soft spectrum]
	\label{prop:soft-spectrum-convergence}
	Fix \(q\in\{1,\dots,s\}\) and \(\beta>0\). Define
	\[
		\widehat{\bm d}^{(q)}
		:=
		\frac{\widetilde{\bm d}^{(q)}}{\|\widetilde{\bm d}^{(q)}\|_{\ell^2}},
		\qquad
		P:=\widehat{\bm d}^{(q)}\otimes \widehat{\bm d}^{(q)*},
		\qquad
		Q:=I-P,
	\]
	and let
	\[
		\widetilde C_{*,\perp}^{(q)}
		:=
		Q\,\widetilde C_*^{(q)}\,Q\big|_{(\widetilde{\bm d}^{(q)})^\perp}.
	\]
	Let
	\[
		\mu_{2,*}^{(q)}\ge \mu_{3,*}^{(q)}\ge\cdots\ge0
	\]
	be the eigenvalues of \(\widetilde C_{*,\perp}^{(q)}\), listed in
	nonincreasing order and counted with multiplicity. Then, for each fixed
	\(n\ge2\),
	\begin{equation}\label{eq:soft-limit}
		\mu_n^{(q)}(\zeta)\longrightarrow \mu_{n,*}^{(q)}
		\qquad\text{as }\zeta\uparrow\zeta_c.
	\end{equation}
\end{proposition}

Proposition~\ref{prop:soft-spectrum-convergence} identifies the limiting
bounded spectrum that remains after the stiff direction is removed. This bounded part is most naturally described at the level of the compressed operator,
rather than through the full diverging spectrum. Figure~\ref{fig:soft-spectrum} illustrates this
soft regime for \(s=3\) and \(s=5\): the upper plots show the convergence of
the first soft eigenvalues as \(\zeta\uparrow\zeta_c\), while the lower plots
display near-critical finite-\(\zeta\) soft spectral profiles across the
symmetry sectors.
Together with Figure~\ref{fig:eigenvalue-trajectories}, this makes the
stiff/soft decomposition visually explicit.

\begin{figure}[t]
	\centering
	\includegraphics[width=\textwidth]{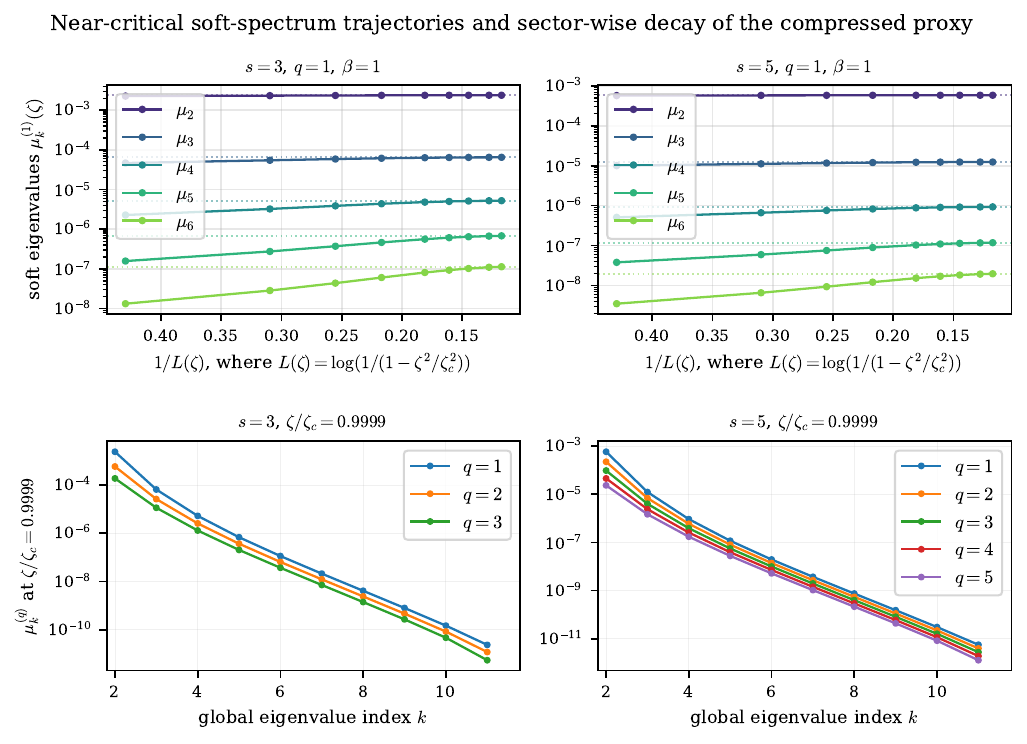}
\caption{
	Soft spectrum after removal of the stiff direction for \(s=3,5\)
	(\(\beta=1\), \(N=40\)).
	Top row: the soft eigenvalues
	\(\mu_2^{(q)}(\zeta),\dots,\mu_6^{(q)}(\zeta)\) in sector \(q=1\),
	plotted against \(1/L(\zeta)\). Their flattening as \(1/L(\zeta)\to0\)
	shows that the soft eigenvalues remain bounded as \(\zeta\uparrow\zeta_c\).
	Bottom row: finite-\(\zeta\) snapshots of the soft eigenvalues for all
	sectors \(q=1,\dots,s\) at \(\zeta/\zeta_c=0.9999\).
}
	\label{fig:soft-spectrum}
\end{figure}

\begin{proof}
	Write $
		\widetilde{\bm d}:=\widetilde{\bm d}^{(q)}$,
		$\widehat{\bm d}:=\widetilde{\bm d}/\|\widetilde{\bm d}\|_{\ell^2}$,
		$
		P:=\widehat{\bm d}\otimes \widehat{\bm d}^{*}$,
		and $Q:=I-P$.
	Let \(\psi_1(\zeta)\) be a normalized eigenvector corresponding to the
	simple eigenvalue \(\mu_1^{(q)}(\zeta)\), and define
	\[
		P_\zeta:=\psi_1(\zeta)\otimes\psi_1(\zeta)^*,
		\qquad
		Q_\zeta:=I-P_\zeta.
	\]
	By \eqref{eq:psi-align}, $
		\|P_\zeta-P\|
		\le
		2\|\psi_1(\zeta)-\widehat{\bm d}\|_{\ell^2}\to0$,
	and therefore
	\begin{equation}\label{eq:proj-conv-soft}
		\|P_\zeta-P\|\to0,
		\qquad
		\|Q_\zeta-Q\|\to0
		\qquad\text{as }\zeta\uparrow\zeta_c.
	\end{equation}

	Since \(P_\zeta\) is the spectral projection of the simple top eigenvalue,
	the restriction of \(\widetilde G^{(q)}(\zeta)\) to \(\Ran Q_\zeta\) has
	eigenvalues precisely $\mu_2^{(q)}(\zeta)\ge \mu_3^{(q)}(\zeta)\ge\cdots$.	It is therefore enough to prove norm convergence of the compressed
	operators:
	\begin{equation}\label{eq:soft-compression-conv}
		\bigl\|
		Q_\zeta\widetilde G^{(q)}(\zeta)Q_\zeta
		-
		Q\widetilde C_*^{(q)}Q
		\bigr\|
		\longrightarrow0.
	\end{equation}

	Using \eqref{eq:G-rankone-tight}, we have
	\[
		Q_\zeta\widetilde G^{(q)}(\zeta)Q_\zeta
		=
		\Gamma^{(q)}\,L(\zeta)\,Q_\zeta P Q_\zeta
		+
		Q_\zeta\widetilde C^{(q)}(\zeta)Q_\zeta.
	\]
	Because \(P=\widehat{\bm d}\otimes\widehat{\bm d}^*\) is rank one and
	\(Q_\zeta=I-P_\zeta\), one has the exact identity
	\[
		Q_\zeta P Q_\zeta
		=
		(Q_\zeta\widehat{\bm d})\otimes(Q_\zeta\widehat{\bm d})^*,
	\]
	hence $
		\|Q_\zeta P Q_\zeta\|
		=
		\|Q_\zeta\widehat{\bm d}\|^2
		=
		1-|\langle\psi_1(\zeta),\widehat{\bm d}\rangle|^2
		=
		\|(I-P)\psi_1(\zeta)\|^2$.
	Therefore \eqref{eq:psi-align-tight} implies
	\begin{equation}\label{eq:spike-vanishes-soft}
		L(\zeta)\,\|Q_\zeta P Q_\zeta\|
		=
		L(\zeta)\,\|(I-P)\psi_1(\zeta)\|^2
		=
		O\bigl(L(\zeta)^{-1}\bigr)
		\longrightarrow0.
	\end{equation}

	Next, $
		\|Q_\zeta \widetilde C^{(q)}(\zeta)Q_\zeta
		-
		Q \widetilde C^{(q)}(\zeta)Q\|
		\le
		2\|\widetilde C^{(q)}(\zeta)\|\,\|Q_\zeta-Q\|
		\to0
	$
	by \eqref{eq:proj-conv-soft}, while
	\[
		\|Q \widetilde C^{(q)}(\zeta)Q-Q\widetilde C_*^{(q)}Q\|
		\le
		\|\widetilde C^{(q)}(\zeta)-\widetilde C_*^{(q)}\|
		\to0
	\]
	by Theorem~\ref{thm:rankone-decomp}. Together with
	\eqref{eq:spike-vanishes-soft}, this yields
	\eqref{eq:soft-compression-conv}.

	Both \(Q_\zeta\widetilde G^{(q)}(\zeta)Q_\zeta\) and
	\(\widetilde C_{*,\perp}^{(q)}\) are compact self-adjoint operators.
	Hence \eqref{eq:soft-compression-conv} and the min--max principle imply
	convergence of every fixed ordered eigenvalue, which is exactly
	\eqref{eq:soft-limit}.
\end{proof}

\begin{remark}[Soft modes]
	\label{rem:soft-modes}
	Proposition~\ref{prop:soft-spectrum-convergence} shows that, after the
	unique stiff direction $\widetilde{\bm d}^{(q)}$ is removed, the remaining
	spectrum has a finite limit as $\zeta\uparrow\zeta_c$. Thus the bounded
	eigenvalues
	\[
		\mu_2^{(q)}(\zeta),\ \mu_3^{(q)}(\zeta),\dots
	\]
	are governed at leading order by the compressed limit operator
	$\widetilde C_{*,\perp}^{(q)}$. In particular, the logarithmic divergence
	of Theorem~\ref{thm:spectral-asymptotics} is entirely carried by the
	rank-one spike.
\end{remark}

Theorem~\ref{thm:rankone-decomp}, Theorem~\ref{thm:spectral-asymptotics}, and
Proposition~\ref{prop:soft-spectrum-convergence} together yield the content of
Theorem~\ref{thm:A} stated in the Introduction.

\FloatBarrier


\section{Scalar Gram continuation beyond the analytic threshold}
\label{sec:regime-2}


Beyond the analytic threshold, three layers of structure must be kept
separate. The intrinsic Hessian coefficients \(H_{mn}\) and the corresponding
unweighted block coefficients remain attached to the conformal map itself and
therefore make sense as long as the geometric map exists. By contrast, the
weighted block operators \(\widetilde G^{(q)}(\zeta)\) from
Section~\ref{sec:regime-I} are a fixed-space Hilbert realization of the
subcritical theory, introduced to isolate the logarithmically singular rank-one
direction from the bounded soft remainder. The scalar quantities
\(\mathcal G_p\), \(\sigma_p\), and \(\rho_p\) reflect the same branch-point
mechanism at coefficient level and are the objects that continue naturally
beyond the weighted operator regime.

Accordingly, the weighted compact-operator framework is not continued past
\(\zeta_c\). As \(\zeta\uparrow\zeta_c\), the Gram weights
\(\sigma_p(\zeta)\) already acquire the borderline logarithmic divergence
isolated in Theorem~\ref{thm:rankone-decomp}. What persists beyond the radius
of convergence is the scalar quantities built from the squared Raney coefficients.
The purpose of this section is to continue those scalar quantities across the slit
plane \(\CC\setminus[\zeta_c^2,\infty)\), identify the branch-point
singularity at \(u=\zeta_c^2\), and recover from it the subcritical
logarithmic divergence.

The section proceeds in three steps. We first identify \(\mathcal G_p(u)\) as a
generalized hypergeometric function and obtain its canonical continuation;
see Proposition~\ref{prop:Gp-hypergeom} and
Corollary~\ref{cor:parametric-excess}. We then analyze the singular point
\(u=\zeta_c^2\) and prove the resonant expansion from
Theorem~\ref{thm:resonant-expansion}. Applying the Euler operator that
reconstructs the Gram weights removes the quadratic prefactor and recovers
the logarithmic divergence from Section~\ref{sec:regime-I}; see
Corollary~\ref{cor:divergence-mechanism}. Finally, the branch-cut jump of
\(\mathcal G_p\) yields a Cauchy--Stieltjes representation, and in the range
\(1\le p\le s\) this representation is positive and hence admits a
Jacobi-matrix realization; see Proposition~\ref{prop:stieltjes},
Proposition~\ref{prop:mu-positive}, and Theorem~\ref{thm:weyl}.

Throughout this section we use the continuation variable \(u=\zeta^2\). In
the Stieltjes/Jacobi interpretation we also use the reciprocal spectral
variable \(t=u^{-1}\). Accordingly, we reserve \(\rho_p(u)\) for the jump
density of the continued Gram weights and \(\varrho_p(t)\) for the
density of the representing measure of \(\mathcal G_p\). For a function
\(f\) defined on the slit plane, we write
\[
	\Disc f(u):=f(u+i0)-f(u-i0),
	\qquad
	u\in(\zeta_c^2,\infty),
\]
for the discontinuity across the branch cut.

\subsection{From Gram weights to scalar generating functions}
\label{subsec:gram-to-gf}

We now isolate the scalar object that carries the singularity. Recall from
Definition~\ref{def:gram-weights} that
\begin{equation}\label{eq:sigma-p-recall}
	\sigma_p(\zeta)
	=
	\|\bm v^{(p)}\|^2
	=
	\sum_{m=0}^{\infty}
	\frac{(p+ms)^2}{p}\,R_{s,p}(m)^2\,\zeta^{2m}.
\end{equation}
For \(\zeta<\zeta_c\) this series converges, while for
\(\zeta\ge\zeta_c\) it diverges. The divergent behavior is entirely encoded
in the squared Raney coefficients, and it is therefore natural to strip away
the explicit polynomial factor \((p+ms)^2/p\).

\begin{definition}[Squared Raney generating function]
	\label{def:Gp}
	For \(p\ge1\), define
	\begin{equation}\label{eq:Gp-def}
		\mathcal G_p(u)
		:=
		\sum_{m=0}^{\infty}R_{s,p}(m)^2\,u^m,
		\qquad
		u\in\CC,
	\end{equation}
	initially in the disk \(|u|<\zeta_c^2\).
\end{definition}

Definition~\ref{def:Gp} removes the explicit polynomial prefactor and leaves
a scalar series whose analytic continuation can be studied directly. The
original Gram weights are recovered from \(\mathcal G_p\) by a fixed
second-order Euler operator.

\begin{lemma}[Euler operator identity]
	\label{lem:euler-identity}
	For \(|u|<\zeta_c^2\),
	\begin{equation}\label{eq:euler-identity}
		\sigma_p(\zeta)
		=
		\frac1p\Bigl(p+s\,u\frac{d}{du}\Bigr)^{\!2}\mathcal G_p(u)
		\Big|_{u=\zeta^2}.
	\end{equation}
\end{lemma}

\begin{proof}
	For each monomial \(u^m\) one has
	\[
		\Bigl(p+s\,u\frac{d}{du}\Bigr)u^m=(p+sm)u^m,
	\]
	hence
	\[
		\frac1p\Bigl(p+s\,u\frac{d}{du}\Bigr)^{\!2}u^m
		=
		\frac{(p+sm)^2}{p}\,u^m.
	\]
	Applying this termwise to the absolutely convergent series
	\eqref{eq:Gp-def} in the disk \(|u|<\zeta_c^2\), and then setting
	\(u=\zeta^2\), gives \eqref{eq:euler-identity}.
\end{proof}

Formula~\eqref{eq:euler-identity} is the basic bridge between the
subcritical operator theory and the continuation theory developed below. The
operator picture detects the singularity through the logarithmic growth of
\(\sigma_p(\zeta)\), while the scalar continuation problem is governed by the
analytic structure of \(\mathcal G_p(u)\) at the branch point
\(u=\zeta_c^2\).


\subsection{Hypergeometric structure and analytic continuation}
\label{subsec:Gp-analytic}

We next identify the analytic representation for \(\mathcal G_p\). First, its radius of convergence is exactly the analytic threshold
\(u=\zeta_c^2\), already encoded in the large-\(m\) asymptotics of the Raney
coefficients.

\begin{proposition}[Radius of convergence]
	\label{prop:Gp-radius}
	The series \eqref{eq:Gp-def} has radius of convergence \(\zeta_c^2\). More
	precisely, as \(m\to\infty\),
	\begin{equation}\label{eq:Rsquared-asymp}
		R_{s,p}(m)^2
		=
		\frac{C_{s,p}}{m^3}\,\zeta_c^{-2m}\bigl(1+o(1)\bigr),
	\end{equation}
	where \(C_{s,p}>0\) depends only on \((s,p)\).
\end{proposition}

\begin{proof}
	By Proposition~\ref{prop:raney-asymp}, equivalently by the squared
	asymptotic formula \eqref{eq:app-Raney-square-asymp} from
	Appendix~\ref{app:raney-asymp}, one has
	\[
R_{s,p}(m)^2\sim C_{s,p}\,m^{-3}\,\zeta_c^{-2m}.
	\]
	This is exactly \eqref{eq:Rsquared-asymp}, and the radius of convergence is
	therefore \(\zeta_c^2\).
\end{proof}

To continue \(\mathcal G_p\) beyond \(|u|<\zeta_c^2\), it is enough to
exploit that the coefficient ratio
\(R_{s,p}(m+1)^2/R_{s,p}(m)^2\) is rational in \(m\). This places
\(\mathcal G_p\) in the generalized hypergeometric class.

\begin{proposition}[Hypergeometric representation]
	\label{prop:Gp-hypergeom}
	For \(|u|<\zeta_c^2\), one has
	\begin{equation}\label{eq:Gp-hypergeom}
		\mathcal G_p(u)
		=
		{}_{2s}F_{2s-1}\!\left(
		\begin{matrix}
			\alpha_1,\ldots,\alpha_{2s} \\
			\beta_1,\ldots,\beta_{2s-1}
		\end{matrix}
		\Big|\,\zeta_c^{-2}u
		\right),
	\end{equation}
	where the parameter multisets are
	\begin{equation}\label{eq:alpha-beta-params}
		\{\alpha_i\}
		=
		2\times\left\{\frac{p+k}{s}:k=0,\ldots,s-1\right\},
		\quad
		\{\beta_j\}
		=
		\{1\}\ \cup\ 2\times\left\{\frac{p+l}{s-1}:l=1,\ldots,s-1\right\}.
	\end{equation}
\end{proposition}

\begin{proof}
	See Appendix~\ref{app:hyp-2}.
\end{proof}

\begin{corollary}[Parametric excess]
	\label{cor:parametric-excess}
	For the hypergeometric parameters
	\eqref{eq:Gp-hypergeom}--\eqref{eq:alpha-beta-params},
	\[
		\gamma:=\sum_j\beta_j-\sum_i\alpha_i=2
	\]
	for every \(s\ge2\) and \(p\ge1\).
\end{corollary}

\begin{proof}
	A direct computation gives
	\[
		\sum_i\alpha_i
		=
		2\sum_{k=0}^{s-1}\frac{p+k}{s}
		=
		2p+(s-1),
	\]
	while
	\[
		\sum_j\beta_j
		=
		1+2\sum_{l=1}^{s-1}\frac{p+l}{s-1}
		=
		1+2p+s.
	\]
	Hence \(\gamma=2\).
\end{proof}

The positivity of the parametric excess has two immediate consequences.
First, the branch point of \(\mathcal G_p\) occurs at
\(\xi:=\zeta_c^{-2}u=1\), that is, at \(u=\zeta_c^2\). Second, the
hypergeometric differential equation is of resonant type there. In
particular, \(\gamma=2\) forces a logarithmic term with prefactor
\((1-\xi)^2\), rather than a pure algebraic singularity.

\begin{lemma}[Reduced hypergeometric parameters after cancellation]
	\label{lem:reduced-hypergeom-data}
	Let \(c_p\) denote the number of cancelled common upper and lower
	parameters in \eqref{eq:Gp-hypergeom}, counted with multiplicity, and set
	\[
		q_p:=2s-1-c_p.
	\]
	After those cancellations, the germ of \(\mathcal G_p\) at \(u=0\) is
	represented by a reduced generalized hypergeometric function of type
	\({}_{q_p+1}F_{q_p}\) with the same parametric excess \(2\). In
	particular, the reduced equation is nontrivial and has finite singular
	points only at \(\xi=0\) and \(\xi=1\), where \(\xi=\zeta_c^{-2}u\).
\end{lemma}

\begin{proof}
	Cancelling a common upper and lower parameter removes the same quantity
	from \(\sum_i\alpha_i\) and \(\sum_j\beta_j\), so the parametric excess is
	unchanged. Since Corollary~\ref{cor:parametric-excess} gives excess \(2\)
	before cancellation, the reduced equation also has excess \(2\).

	If the reduction were of type \({}_1F_0\), its unique upper parameter
	would have to equal \(-2\) in order to have excess \(2\), which is
	impossible because every surviving upper parameter comes from the positive
	list \eqref{eq:alpha-beta-params}. Thus \(q_p\ge1\), so the reduced form is
	genuinely of type \({}_{q_p+1}F_{q_p}\). For such equations the only finite
	singular points are \(0\) and \(1\) in the variable \(\xi\); see
	\cite[\S16.8, \S16.11]{NIST:DLMF}.
\end{proof}

\begin{corollary}[Analytic continuation]
	\label{cor:Gp-continuation}
	The function \(\mathcal G_p\) extends to a single-valued holomorphic
	function on the slit plane
	\[
		\CC\setminus[\zeta_c^2,\infty).
	\]
\end{corollary}

\begin{proof}
	By Lemma~\ref{lem:reduced-hypergeom-data}, after cancellation of common
	upper and lower parameters the germ at \(u=0\) is represented by a reduced
	generalized hypergeometric function of type \({}_{q_p+1}F_{q_p}\) with the
	same parametric excess \(2\). For that reduced hypergeometric differential
	equation, the only finite singular points in the variable
	\(\xi=\zeta_c^{-2}u\) are \(\xi=0\) and \(\xi=1\). Equivalently, the only
	finite singular points in the variable \(u\) are \(u=0\) and
	\(u=\zeta_c^2\); see \cite[\S16.8, \S16.11]{NIST:DLMF}. Hence the germ at
	the origin continues uniquely along every path in
	\(\CC\setminus[\zeta_c^2,\infty)\), which is simply connected. This yields
	a single-valued holomorphic continuation of \(\mathcal G_p\) to the slit
	plane.
\end{proof}

\subsection{Resonant expansion at the branch point}
\label{subsec:branch-point}

We now analyze the singularity of \(\mathcal G_p\) at the endpoint
\(u=\zeta_c^2\) of the disk of convergence. Because the parametric excess is
the integer \(\gamma=2\), the local continuation at the branch point is
resonant, and the first possible logarithmic term occurs at order
\((1-u/\zeta_c^2)^2\log(1-u/\zeta_c^2)\).

\begin{lemma}[Reduced local description at the branch point]
	\label{lem:xi-one-exponents}
	Let \(\xi=\zeta_c^{-2}u\), and write the germ of \(\mathcal G_p\) at
	\(\xi=0\) in its reduced hypergeometric form \({}_{q_p+1}F_{q_p}\),
	obtained by cancelling any common upper and lower parameters in
	\eqref{eq:Gp-hypergeom}. Then \(\xi=1\) is a regular singular point of the
	reduced differential equation. Its local exponents are
	\[
		\{0,2\}\quad\text{if }q_p=1,
		\qquad
		\{0,1,2\}\quad\text{if }q_p=2,
		\qquad
		\{0,1,\dots,q_p-1\}\cup\{2\}\quad\text{if }q_p\ge3,
	\]
	where in the last case the exponent \(2\) has multiplicity two. In
	particular, the exponent \(0\) is simple, while the exponent \(1\) is
	either absent or simple. Hence no terms of the form \(\log(1-\xi)\) or
	\((1-\xi)\log(1-\xi)\) can occur in the local continuation of the germ.
\end{lemma}

\begin{proof}
	By Lemma~\ref{lem:reduced-hypergeom-data}, after cancellation
	\(\mathcal G_p\) satisfies a reduced generalized hypergeometric equation of
	type \({}_{q_p+1}F_{q_p}\) with parametric excess \(\gamma=2\). For such an
	equation, the local exponents at \(\xi=1\) are
	\(0,1,\dots,q_p-1,\gamma\); see \cite[\S16.8]{NIST:DLMF}. Substituting
	\(\gamma=2\) gives the stated list, with exponent \(2\) repeated when
	\(q_p\ge3\). The exponent \(0\) is therefore simple, and the exponent
	\(1\) is either absent or simple. Consequently, the local continuation of
	the germ selected by the Taylor series at \(\xi=0\) cannot contain
	logarithmic terms of order \((1-\xi)^0\) or \((1-\xi)^1\). The first
	order at which a logarithmic term may occur is thus \((1-\xi)^2\).
\end{proof}

\begin{theorem}[Resonant expansion]
	\label{thm:resonant-expansion}
	There exist functions \(A(u)\) and \(B(u)\), analytic in a neighborhood of
	\(u=\zeta_c^2\), such that
	\begin{equation}\label{eq:resonant-expansion}
		\mathcal G_p(u)
		=
		A(u)+
		B(u)\Bigl(1-\frac{u}{\zeta_c^2}\Bigr)^2
		\log\Bigl(1-\frac{u}{\zeta_c^2}\Bigr),
	\end{equation}
	for all \(u\) in some slit neighborhood of \(\zeta_c^2\) contained in
	\(\CC\setminus[\zeta_c^2,\infty)\). Moreover,
	\[
		B(\zeta_c^2)<0
	\]
	for all \(s\ge2\) and \(p\ge1\).
\end{theorem}

\begin{proof}
	Set \(\xi:=u/\zeta_c^2\). By Corollary~\ref{cor:Gp-continuation}, the Taylor
	germ at \(\xi=0\) extends holomorphically to
	\(\CC\setminus[1,\infty)\). By Lemma~\ref{lem:reduced-hypergeom-data}, after
	cancellation of common upper and lower parameters in
	\eqref{eq:Gp-hypergeom}, this continuation is a reduced
	\({}_{q_p+1}F_{q_p}\)-function with parametric excess
	\[
		\gamma=2.
	\]

	Hence \(\xi=1\) is a regular singular point of the corresponding differential
	equation. By Lemma~\ref{lem:xi-one-exponents}, the exponents \(0\) and \(1\)
	at \(\xi=1\) are simple, so no logarithmic terms of order
	\((1-\xi)^0\) or \((1-\xi)^1\) can occur in the continuation of the
	Taylor germ from \(\xi=0\). The remaining Frobenius sectors, when present,
	start at higher integer powers and are analytic at \(\xi=1\). For a
	reduced \({}_{q_p+1}F_{q_p}\)-equation with positive integer parametric
	excess \(m\), the analytic continuation of the Taylor germ at \(\xi=0\) has
	a local representation near \(\xi=1\) of the form
	\[
		(1-\xi)^m\bigl(C_0(\xi)+D_0(\xi)\log(1-\xi)\bigr),
	\]
	where \(C_0\) and \(D_0\) are analytic near \(\xi=1\); see
	\cite[\S16.8(i), \S16.11(ii)]{NIST:DLMF} and \cite[Ch.~4]{Slater1966}.
	Since here \(m=\gamma=2\), we obtain on a slit neighborhood of \(\xi=1\)
	\[
		\mathcal G_p(\zeta_c^2\xi)
		=
		A_0(\xi)+B_0(\xi)(1-\xi)^2\log(1-\xi),
	\]
	where \(A_0\) and \(B_0\) are analytic near \(\xi=1\). The simplicity of the
	exponents \(0\) and \(1\) excludes lower-order logarithmic terms. Returning to \(u=\zeta_c^2\xi\), define
	\[
		A(u):=A_0(u/\zeta_c^2),
		\qquad
		B(u):=B_0(u/\zeta_c^2).
	\]
	This gives \eqref{eq:resonant-expansion}. The sign of \(B(\zeta_c^2)\) is
	computed independently in Proposition~\ref{prop:B1-closed}, which gives
	\[
		B(\zeta_c^2)
		=
		-
		\frac{p^2}{4\pi}\,\frac{s^{2p-1}}{(s-1)^{2p+1}}
		<0.
	\]
\end{proof}

\begin{corollary}[Mechanism of Gram weights divergence]
	\label{cor:divergence-mechanism}
	As \(\zeta\uparrow\zeta_c\), equivalently \(u=\zeta^2\uparrow\zeta_c^2\),
	the Gram weights satisfy
	\[
		\sigma_p(\zeta)
		=
		\widetilde A_p
		+
		\frac{2s^2}{p}\,B(\zeta_c^2)\,
		\log\Bigl(1-\frac{\zeta^2}{\zeta_c^2}\Bigr)
		+
		o(1),
	\]
	where \(\widetilde A_p\in\mathbb R\), and \(B(\zeta_c^2)<0\) is the coefficient
	from Theorem~\textup{\ref{thm:resonant-expansion}}. In particular,
	\[
		\sigma_p(\zeta)\to+\infty
	\]
	logarithmically, in agreement with
	Proposition~\textup{\ref{prop:gram-divergence}}.
\end{corollary}

\begin{proof}
	Apply Lemma~\ref{lem:euler-identity} to the expansion
	\eqref{eq:resonant-expansion}, with \(u=\zeta^2\), and write $w:=1-u/\zeta_c^2$.
	Since \(A(u)\) and \(B(u)\) are analytic at \(u=\zeta_c^2\), one has
	\[
		\mathcal G_p(u)
		=
		A(u)+B(\zeta_c^2)\,w^2\log w+O(w^3\log w).
	\]
	The Euler operator \(\bigl(p+s\,u\frac{d}{du}\bigr)^2\) preserves analyticity,
	so the contribution of \(A(u)\) is a constant plus \(o(1)\), while
	\(O(w^3\log w)\) contributes only \(o(1)\). For the singular term,
	\begin{equation}\label{eq:Euler-on-w2logw}
		\Bigl(p+s\,u\frac{d}{du}\Bigr)^2\!\bigl[w^2\log w\bigr]
		=
		2s^2\log w+3s^2+O\bigl(w\log w\bigr),
		\qquad w\downarrow0.
	\end{equation}
	Hence
	\[
		\sigma_p(\zeta)
		=
		\widetilde A_p
		+
		\frac{2s^2}{p}\,B(\zeta_c^2)\log w
		+
		o(1),
	\]
	for some \(\widetilde A_p\in\mathbb R\). Since
	\(w=1-\zeta^2/\zeta_c^2\), this is the claimed expansion. Finally,
	\(B(\zeta_c^2)<0\) and \(\log(1-\zeta^2/\zeta_c^2)\to-\infty\), hence
	\(\sigma_p(\zeta)\to+\infty\).
\end{proof}
\subsection{Continued Gram weights and edge density}
\label{subsec:continued-weights}

Once \(\mathcal G_p\) has been continued to the slit plane, the Euler
identity from Lemma~\ref{lem:euler-identity} gives a canonical continuation of
the Gram weights as well.

\begin{definition}[Continued Gram weights]
	\label{def:continued-weights}
	For \(u\in\CC\setminus[\zeta_c^2,\infty)\), define
	\begin{equation}\label{eq:sigma-tilde}
		\sigma_p^{\mathrm{cont}}(u)
		:=
		\frac1p\Bigl(p+s\,u\frac{d}{du}\Bigr)^{\!2}\mathcal G_p(u).
	\end{equation}
	For \(|u|<\zeta_c^2\), this agrees with the original Gram weights
	\(\sigma_p(\zeta)\) under the substitution \(u=\zeta^2\), by
	Lemma~\ref{lem:euler-identity}.
\end{definition}

The function \(\sigma_p^{\mathrm{cont}}\) inherits the branch cut
\([\zeta_c^2,\infty)\) from \(\mathcal G_p\). Its jump across the cut is the
natural supercritical analogue of the logarithmic divergence on the
subcritical side.

\begin{definition}[Discontinuity density]
	\label{def:discontinuity}
	For \(u>\zeta_c^2\), define
	\begin{equation}\label{eq:rho-def}
		\rho_p(u)
		:=
		-\frac{1}{2\pi i}\,\Disc\,\sigma_p^{\mathrm{cont}}(u)
		=
		-\frac{1}{2\pi i}
		\bigl(\sigma_p^{\mathrm{cont}}(u+i0)-\sigma_p^{\mathrm{cont}}(u-i0)\bigr).
	\end{equation}
\end{definition}

The sign convention in \eqref{eq:rho-def} is chosen so that the edge value is
positive when \(B(\zeta_c^2)<0\), in agreement with the standard resolvent
orientation.

\begin{theorem}[Nonvanishing edge density]
	\label{thm:edge-density}
	The limit
	\[
		\rho_p(\zeta_c^2)
		:=
		\lim_{u\downarrow\zeta_c^2}\rho_p(u)
	\]
	exists and is nonzero for all \(s\ge2\) and \(p\ge1\). More precisely,
	\begin{equation}\label{eq:rho-edge}
		\rho_p(\zeta_c^2)
		=
		-\frac{2s^2}{p}\,B(\zeta_c^2).
	\end{equation}
\end{theorem}

\begin{proof}
	Let \(u>\zeta_c^2\) and set $w:=1- u / \zeta_c^2\in(-\infty,0)$.
	By Theorem~\ref{thm:resonant-expansion}, in a punctured neighborhood of
	\(u=\zeta_c^2\) one has
	\[
		\mathcal G_p(u)=A(u)+B(u)\,w^2\log w,
	\]
	with \(A\) and \(B\) analytic near \(\zeta_c^2\). Since \(A\) and \(B\) are
	single-valued across the cut, only the logarithm contributes to the jump.
	On the principal branch, $\Disc\,\log w = 2\pi i$, $u>\zeta_c^2$,
	hence
	\begin{equation}\label{eq:disc-Gp-local}
		\Disc\,\mathcal G_p(u)=2\pi i\,B(u)\,w^2,
		\qquad u>\zeta_c^2 \text{ close to }\zeta_c^2.
	\end{equation}

	Because the Euler operator has analytic coefficients away from the endpoint,
	it may be applied separately to the upper and lower holomorphic boundary
	values on the cut. Subtracting the two resulting expressions shows that it
	commutes with the discontinuity operation. Therefore,
	\[
		\Disc\,\sigma_p^{\mathrm{cont}}(u)
		=
		\frac1p\Bigl(p+s\,u\frac{d}{du}\Bigr)^{\!2}\Disc\,\mathcal G_p(u)
		=
		\frac{2\pi i}{p}
		\Bigl(p+s\,u\frac{d}{du}\Bigr)^{\!2}\!\bigl[B(u)w^2\bigr].
	\]

	Now \(u=\zeta_c^2(1-w)\), so $
		u(d/du)=-(1-w)(d/dw)$.
	Since \(B(u)=B(\zeta_c^2)+O(w)\), one finds
	\[
		\Bigl(p+s\,u\frac{d}{du}\Bigr)\!\bigl[B(u)w^2\bigr]
		=
		-2s\,B(\zeta_c^2)\,w+O(w^2),
	\]
	and applying the same operator once more yields
	\begin{equation}\label{eq:Euler-on-w2}
		\Bigl(p+s\,u\frac{d}{du}\Bigr)^{\!2}\!\bigl[B(u)w^2\bigr]
		=
		2s^2\,B(\zeta_c^2)+O(w),
		\qquad u\to\zeta_c^2.
	\end{equation}
	Substituting into \eqref{eq:rho-def} gives
	\[
		\rho_p(u)
		=
		-\frac1p\Bigl(2s^2\,B(\zeta_c^2)+O(w)\Bigr),
	\]
	and letting \(u\downarrow\zeta_c^2\) proves \eqref{eq:rho-edge}.
\end{proof}

\begin{remark}[Matching across the analytic threshold]
	\label{rem:edge-matching}
	Corollary~\ref{cor:divergence-mechanism} and
	Theorem~\ref{thm:edge-density} describe the same resonant singularity from
	opposite sides of the threshold. On the subcritical side
	\((u\uparrow\zeta_c^2)\), the Euler operator converts the term
	\[
		B(u)\Bigl(1-\frac{u}{\zeta_c^2}\Bigr)^2
		\log\Bigl(1-\frac{u}{\zeta_c^2}\Bigr)
	\]
	into a bare logarithm, producing the divergence
	\(\sigma_p(\zeta)\to+\infty\). On the supercritical side
	\((u>\zeta_c^2)\), the same logarithm acquires the jump
	\[
		\Disc\,\log\Bigl(1-\frac{u}{\zeta_c^2}\Bigr)=2\pi i,
	\]
	and this yields the nonzero edge value \(\rho_p(\zeta_c^2)\). Thus the
	subcritical blow-up and the supercritical edge density are two boundary
	manifestations of one and the same branch-point coefficient
	\(B(\zeta_c^2)\).
\end{remark}

To illustrate the continuation across the critical point, we plot in
Figure~\ref{fig:gp-rho} the continued Gram weights
\(\sigma_p^{\mathrm{cont}}(u)\), which agrees with the subcritical Gram
weights for \(u<\zeta_c^2\), together with the supercritical discontinuity
density \(\rho_p(u)\) for several values of \(p\).
The lower plots show that the edge value
\[
\rho_p(\zeta_c^2)= -\frac{2s^2}{p}\,B(\zeta_c^2)
\]
is positive, in agreement with the local expansion at the branch point.
At the same time, the additional curves with larger \(p\) show that
\(\rho_p(u)\) need not remain positive for all \(u>\zeta_c^2\): after
starting from a positive edge value, the density may cross zero and become
negative further along the supercritical branch.
Thus the sign of \(\rho_p\) is determined locally near \(u=\zeta_c^2\), but need not remain fixed throughout the whole continuation domain.

\begin{figure}[t]
	\centering
	\includegraphics[width=\textwidth]{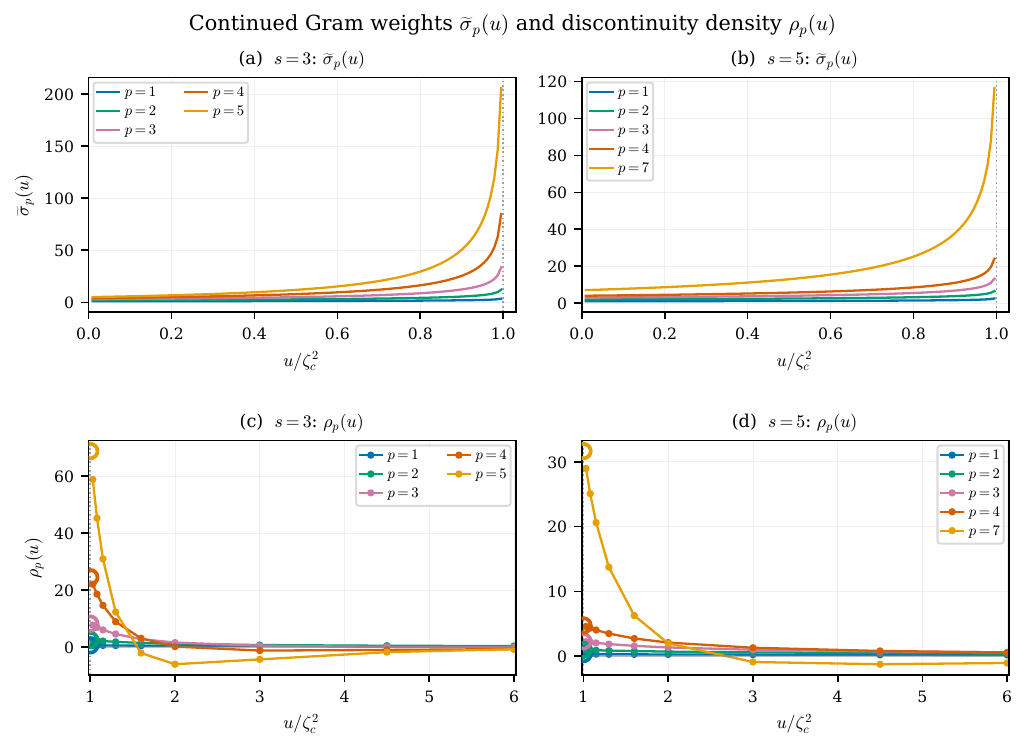}
	\caption{
	Continued Gram weights \(\sigma_p^{\mathrm{cont}}(u)\) and discontinuity density
	\(\rho_p(u)\) for \(s=3,5\).
	Top row: the analytically continued Gram weights
	\(\sigma_p^{\mathrm{cont}}(u)\) on the subcritical side \(0<u<\zeta_c^2\).
	Bottom row: the discontinuity density \(\rho_p(u)\) on the supercritical
	side \(u>\zeta_c^2\).
	The open circles at \(u=\zeta_c^2\) mark the edge values $
	\rho_p(\zeta_c^2)>0$.
	For larger values of \(p\), the density becomes negative away from the
	edge, showing that positivity at \(u=\zeta_c^2\) does not persist on the
	entire supercritical branch.
	}
	\label{fig:gp-rho}
\end{figure}

This numerical behavior is consistent with the edge asymptotics proved above,
while showing that the positivity statement is local in \(u\) and does not
extend to the full supercritical branch.

\subsection{Cauchy--Stieltjes representation}
\label{subsec:stieltjes}

The hypergeometric continuation of \(\mathcal G_p\) on the slit plane
produces canonical boundary values across the cut and therefore a
Cauchy--Stieltjes representation. This representation does not require any
positivity assumption and is valid for every \(p\ge1\).

\begin{lemma}[Growth at infinity]
	\label{lem:Gp-infty}
	Fix \(p\ge1\). As \(|u|\to\infty\) in
	\(\CC\setminus[\zeta_c^2,\infty)\),
	\begin{equation}\label{eq:Gp-infty}
		\mathcal G_p(u)
		=
		O\bigl(|u|^{-p/s}(1+\log|u|)\bigr).
	\end{equation}
	Consequently, for
	\[
		F_p(\eta):=\eta^{-1}\mathcal G_p(1/\eta),
	\]
	one has
	\begin{equation}\label{eq:Fp-zero}
		F_p(\eta)
		=
		O\bigl(|\eta|^{p/s-1}(1+\log(1/|\eta|))\bigr),
		\qquad
		\eta\to0,\ \eta\notin[0,1/\zeta_c^2].
	\end{equation}
\end{lemma}

\begin{proof}
	By Proposition~\ref{prop:Gp-hypergeom},
	\(\mathcal G_p(u)\) is a generalized hypergeometric function in the
	variable \(\xi=u/\zeta_c^2\), with numerator parameters
	\[
		\alpha_k=\frac{p+k}{s},
		\qquad k=0,\dots,s-1,
	\]
	each occurring with multiplicity \(2\). The generalized hypergeometric
	differential equation has regular singular points at \(\xi=0,1,\infty\),
	and the local exponents at \(\xi=\infty\) are precisely the numerator
	parameters; see \cite[\S16.11(i)]{NIST:DLMF} and \cite[Ch.~4]{Slater1966}.
	Because each exponent appears with multiplicity \(2\), the local basis at
	\(\infty\) contains at most one logarithm for each exponent, so the
	asymptotic expansion takes the form
	\[
		\mathcal G_p(u)
		=
		\sum_{k=0}^{s-1}
		u^{-\alpha_k}\bigl(A_k\log u+B_k\bigr)
		+
		O\!\bigl(|u|^{-(p+1)/s}\log|u|\bigr),
	\]
	uniformly on every closed subsector of the slit plane. The smallest
	exponent is \(\alpha_0=p/s\), which yields \eqref{eq:Gp-infty}.
	Substituting \(u=1/\eta\) gives \eqref{eq:Fp-zero}.
\end{proof}

\begin{proposition}[Stieltjes representation]
	\label{prop:stieltjes}
	There exists a finite real signed Borel measure \(\nu_p\) of bounded total
	variation, supported on \([0,1/\zeta_c^2]\), such that
	\begin{equation}\label{eq:stieltjes-Gp}
		\mathcal G_p(u)
		=
		\int_0^{1/\zeta_c^2}\frac{d\nu_p(t)}{1-ut},
		\qquad
		u\in\CC\setminus[\zeta_c^2,\infty).
	\end{equation}
	Moreover,
	\[
		\nu_p\bigl([0,1/\zeta_c^2]\bigr)=\mathcal G_p(0)=1,
	\]
	and \(\nu_p\) is absolutely continuous on the open interval
	\((0,1/\zeta_c^2)\), with density
	\begin{equation}\label{eq:mu-jump}
		\frac{d\nu_p}{dt}(t)
		=
		\frac{1}{2\pi i\,t}\,
		\Disc\,\mathcal G_p\Bigl(\frac{1}{t}\Bigr),
		\qquad
		0<t<1/\zeta_c^2.
	\end{equation}
\end{proposition}

\begin{proof}
	Set
	\[
		F_p(\eta):=\eta^{-1}\mathcal G_p(1/\eta).
	\]
	Since \(\mathcal G_p\) is holomorphic on
	\(\CC\setminus[\zeta_c^2,\infty)\), the function \(F_p\) is holomorphic on
	\(\CC\setminus[0,1/\zeta_c^2]\). From the Taylor expansion
	\(\mathcal G_p(u)=1+O(u)\) at \(u=0\), one obtains $F_p(\eta)=\eta^{-1}+O(\eta^{-2})$, as $|\eta|\to\infty$. On the other hand, Lemma~\ref{lem:Gp-infty} gives $F_p(\eta)=o(|\eta|^{-1})$,
	as	$\eta\to0$,	away from the cut.

	Let \(\eta\in\CC\setminus[0,1/\zeta_c^2]\), and apply Cauchy’s theorem to a
	keyhole contour around the interval \([0,1/\zeta_c^2]\), with outer radius
	\(R\) and inner radius \(\varepsilon\). The contribution of the outer circle
	vanishes as \(R\to\infty\) because \(F_p(\xi)=O(\xi^{-1})\), while the
	inner-circle contribution vanishes as \(\varepsilon\to0\) by the bound
	\eqref{eq:Fp-zero}. Passing to the limit yields
	\begin{equation}\label{eq:keyhole-Fp}
		F_p(\eta)
		=
		\frac{1}{2\pi i}\int_0^{1/\zeta_c^2}
		\frac{\Disc\,F_p(t)}{t-\eta}\,dt.
	\end{equation}

	Define a signed Borel measure on \([0,1/\zeta_c^2]\) by
	\[
		d\nu_p(t):=\frac{1}{2\pi i}\,\Disc\,F_p(t)\,dt.
	\]
	Its total variation is finite, and this is exactly where the endpoint
	estimates enter. If \(0<t<1/\zeta_c^2\) is close to the right endpoint and
	\(u=1/t\), then \(u>\zeta_c^2\) and Theorem~\ref{thm:resonant-expansion}
	gives
	\[
		\Disc\,\mathcal G_p(u)
		=
		2\pi i\,B(u)\Bigl(1-\frac{u}{\zeta_c^2}\Bigr)^2.
	\]
	Hence
	\[
		\Disc\,F_p(t)
		=
		t^{-1}\Disc\,\mathcal G_p(1/t)
		=
		O\Bigl(\Bigl(\frac1{\zeta_c^2}-t\Bigr)^2\Bigr),
		\qquad t\uparrow \frac1{\zeta_c^2},
	\]
	so the density is integrable at the right endpoint. Near \(t=0\),
	\eqref{eq:Fp-zero} gives
	\[
		\Disc\,F_p(t)=O\bigl(t^{p/s-1}(1+|\log t|)\bigr),
	\]
	which is integrable because \(p/s>0\). Therefore \(|\Disc F_p(t)|\) is
	integrable on \([0,1/\zeta_c^2]\), and the resulting measure \(\nu_p\) has
	bounded total variation. Thus \eqref{eq:keyhole-Fp} becomes
	\[
		F_p(\eta)=\int_0^{1/\zeta_c^2}\frac{d\nu_p(t)}{t-\eta}.
	\]
	Substituting \(\eta=1/u\) and multiplying by \(u\) gives
	\eqref{eq:stieltjes-Gp}.

	Next,
	\[
		\Disc\,F_p(t)=t^{-1}\Disc\,\mathcal G_p(1/t),
	\]
	which yields the jump formula \eqref{eq:mu-jump}. Evaluating
	\eqref{eq:stieltjes-Gp} at \(u=0\) gives
	\[
		\nu_p\bigl([0,1/\zeta_c^2]\bigr)=\mathcal G_p(0)=1.
	\]
	Finally, \(\nu_p\) is real because \(\mathcal G_p\) has real Taylor
	coefficients and therefore satisfies Schwarz reflection on the slit plane.
\end{proof}

\begin{remark}
	\label{rem:stieltjes-general}
	For general \(p\), the representation
	\eqref{eq:stieltjes-Gp}--\eqref{eq:mu-jump} is used only as an analytic
	description of the continued scalar Gram quantities. Positivity is not required at
	this stage. The operator-theoretic interpretation enters only later, in the
	range \(1\le p\le s\), where \(\nu_p\) becomes a positive measure and one can
	pass to the Jacobi/OPRL framework.
\end{remark}

\subsection{Jacobi realization in the positive range}
\label{subsec:jacobi}

For general \(p\), Proposition~\ref{prop:stieltjes} yields only a signed
representing measure for \(\mathcal G_p\). We now restrict to the positive
range \(1\le p\le s\), where positivity is supplied by an external
Hausdorff-moment theorem for Raney numbers. Before passing to the Jacobi/OPRL framework, we state only the part needed here, in the normalization used in this paper.

\begin{lemma}[External Hausdorff moment criterion in the present normalization]
	\label{lem:raney-moment-input}
	Assume \(s\ge2\) and \(0\le p\le s\). Then there exists a probability
	measure \(\nu_{s,p}\) supported on
	\[
		[0,\tau_s],
		\qquad
		\tau_s:=\frac{s^s}{(s-1)^{s-1}}=\zeta_c^{-1},
	\]
	such that
	\[
		R_{s,p}(n)=\int_0^{\tau_s} x^n\,d\nu_{s,p}(x),
		\qquad n\ge0.
	\]
	Equivalently, the rescaled sequence
	\(\{R_{s,p}(n)\zeta_c^n\}_{n\ge0}\) is a Hausdorff moment sequence on
	\([0,1]\).

	This is the only lemma that relies on positivity of the external Raney moments. In what follows, we use only the existence of the positive
	measure and the support endpoint \(\tau_s=\zeta_c^{-1}\); no explicit
	formula for the density is needed. These facts are supplied by
	\cite[Thm.~5]{LiuPego2016} together with
	\cite{LiuPegoCorrigendum2026}. For background on the corresponding Raney
	distributions, see also \cite{ForresterLiu2015}.
\end{lemma}

\begin{proof}
	By Proposition~\ref{prop:raney-lagrange},
	\[
		R_{s,p}(n)=\frac{p}{sn+p}\binom{sn+p}{n}.
	\]
	In the notation of \cite{LiuPego2016}, this is exactly the Raney sequence
	\(A_n(s,p)\). Therefore \cite[Thm.~5]{LiuPego2016}, as corrected in
	\cite{LiuPegoCorrigendum2026}, applies with the parameter pair
	\((p,r)=(s,p)\): the hypotheses are satisfied because \(s\ge1\) and
	\(0\le p\le s\). It follows that \(\{R_{s,p}(n)\}_{n\ge0}\) is represented
	by a probability measure supported on \([0,\tau_s]\), with
	\[
		\tau_s=\frac{s^s}{(s-1)^{s-1}}.
	\]
	In the normalization of the present paper,
	\(\tau_s=\zeta_c^{-1}\), which gives exactly the stated moment
	representation. After the rescaling \(x=\tau_s y\), this is equivalent to
	the Hausdorff moment statement for
	\(\{R_{s,p}(n)\zeta_c^n\}_{n\ge0}\) on \([0,1]\).
\end{proof}

\begin{proposition}[Positivity for \(1\le p\le s\)]
	\label{prop:mu-positive}
	Assume \(1\le p\le s\). Then the measure \(\nu_p\) in
	Proposition~\ref{prop:stieltjes} can be chosen to be a positive probability
	measure supported on \([0,1/\zeta_c^{2}]\). Equivalently, after the
	rescaling \(t\mapsto t\,\zeta_c^{2}\), the sequence
	\[
		\{R_{s,p}(n)^2\,\zeta_c^{2n}\}_{n\ge0}
	\]
	is a Hausdorff moment sequence.
\end{proposition}

\begin{proof}
	By Lemma~\ref{lem:raney-moment-input}, for \(1\le p\le s\) there exists a
	probability measure \(\nu_{s,p}\) supported on
	\([0,\tau_s]=[0,\zeta_c^{-1}]\) such that
	\[
		R_{s,p}(n)=\int_0^{\tau_s}x^n\,d\nu_{s,p}(x),
		\qquad n\ge0.
	\]

	Let \(\widetilde\nu_p\) be the pushforward of
	\(\nu_{s,p}\otimes\nu_{s,p}\) under the multiplication map
	\[
		(x,y)\mapsto xy.
	\]
	Then \(\widetilde\nu_p\) is a probability measure supported on
	\([0,\tau_s^2]=[0,\zeta_c^{-2}]\), and for every \(n\ge0\),
	\[
		\int_0^{1/\zeta_c^2} t^n\,d\widetilde\nu_p(t)
		=
		\iint (xy)^n\,d\nu_{s,p}(x)\,d\nu_{s,p}(y)
		=
		\biggl(\int_0^{\tau_s}x^n\,d\nu_{s,p}(x)\biggr)^2
		=
		R_{s,p}(n)^2.
	\]
	Thus the rescaled sequence
	\(\{R_{s,p}(n)^2\zeta_c^{2n}\}_{n\ge0}\) is a Hausdorff moment sequence on
	\([0,1]\). Consequently,
	\[
		\int_0^{1/\zeta_c^2}\frac{d\widetilde\nu_p(t)}{1-ut}
		=
		\sum_{n\ge0}R_{s,p}(n)^2\,u^n
		=
		\mathcal G_p(u),
		\qquad |u|<\zeta_c^2.
	\]
	Since both sides are holomorphic on
	\(\CC\setminus[\zeta_c^2,\infty)\), the identity theorem extends this
	representation to the whole slit plane.

	On the other hand, Proposition~\ref{prop:stieltjes} already provides a
	finite signed measure \(\nu_p\) on \([0,1/\zeta_c^2]\) with the same
	Stieltjes transform:
	\[
		\mathcal G_p(u)
		=
		\int_0^{1/\zeta_c^2}\frac{d\nu_p(t)}{1-ut},
		\qquad
		u\in\CC\setminus[\zeta_c^2,\infty).
	\]
	Hence
	\[
		\int_0^{1/\zeta_c^2}\frac{d(\nu_p-\widetilde\nu_p)(t)}{1-ut}=0
	\]
	on the slit plane. By uniqueness of the Cauchy--Stieltjes transform of a
	finite compactly supported measure, one has
	\(\nu_p=\widetilde\nu_p\). Therefore the measure in
	Proposition~\ref{prop:stieltjes} is in fact positive, and being a
	probability measure it has total mass \(1\).
\end{proof}

Once positivity is available, the standard orthogonal-polynomial machinery
produces a Jacobi operator with spectral measure \(\nu_p\).

\begin{definition}[Jacobi operator associated with \(\mathcal G_p\)]
	\label{def:jacobi}
	Assume \(1\le p\le s\), and let \(\{P_n^{(p)}\}_{n\ge0}\) be the monic
	orthogonal polynomials with respect to the positive measure \(\nu_p\) from
	Proposition~\ref{prop:mu-positive}. Then there exist coefficients
	\[
		a_n^{(p)}>0,
		\qquad
		b_n^{(p)}\in\R,
	\]
	such that
	\begin{equation}\label{eq:three-term}
		tP_n^{(p)}(t)
		=
		P_{n+1}^{(p)}(t)+b_n^{(p)}P_n^{(p)}(t)+(a_n^{(p)})^2P_{n-1}^{(p)}(t),
	\end{equation}
	with \(P_{-1}^{(p)}\equiv0\). The associated Jacobi operator \(J_p\) is the
	bounded self-adjoint operator on \(\ell^2(\NN_0)\) with tridiagonal matrix
	\[
		J_p=
		\begin{pmatrix}
			b_0^{(p)} & a_1^{(p)} & 0         & \cdots \\
			a_1^{(p)} & b_1^{(p)} & a_2^{(p)} & \ddots \\
			0         & a_2^{(p)} & b_2^{(p)} & \ddots \\
			\vdots    & \ddots    & \ddots    & \ddots
		\end{pmatrix}.
	\]
	Since \(\textrm{supp}(\nu_p)\subset[0,1/\zeta_c^2]\), one has
	\[
		\sigma(J_p)\subset[0,1/\zeta_c^2].
	\]
\end{definition}

\begin{theorem}[Weyl function identity]
	\label{thm:weyl}
	Assume \(1\le p\le s\). Then \(\mathcal G_p\) is the Weyl \(m\)-function of
	\(J_p\) in the reciprocal spectral parameter:
	\begin{equation}\label{eq:weyl}
		\mathcal G_p(u)
		=
		\langle e_0,(I-uJ_p)^{-1}e_0\rangle,
		\qquad
		u\in\CC\setminus[\zeta_c^2,\infty),
	\end{equation}
	where \(e_0=(1,0,0,\dots)^T\). Equivalently, for \(u\neq0\), the formula
	holds whenever \(u^{-1}\in\rho(J_p)\).
\end{theorem}

\begin{proof}
	By the spectral theorem,
	\[
		\langle e_0,(I-uJ_p)^{-1}e_0\rangle
		=
		\int_0^{1/\zeta_c^2}\frac{d\nu_p(t)}{1-ut},
		\qquad
		u\in\CC\setminus[\zeta_c^2,\infty),
	\]
	where \(\nu_p\) is the spectral measure of \(J_p\) at \(e_0\). By
	Proposition~\ref{prop:mu-positive}, this is precisely the positive
	representing measure of \(\mathcal G_p\), and the right-hand side equals
	\(\mathcal G_p(u)\) by \eqref{eq:stieltjes-Gp}. This proves
	\eqref{eq:weyl}; see, for example, \cite[Ch.~III]{Akhiezer1965}.
\end{proof}

\begin{theorem}[Stieltjes--Perron inversion]
	\label{thm:spectral-density}
	Assume \(1\le p\le s\). Then the positive measure \(\nu_p\) from
	Proposition~\ref{prop:mu-positive} is absolutely continuous on
	\((0,1/\zeta_c^2)\), and its density is
	\begin{equation}\label{eq:perron}
		\varrho_p(t)
		=
		\frac{1}{\pi t}\,\Im\,\mathcal G_p\Bigl(\frac{1}{t}+i0\Bigr)
		=
		\frac{1}{2\pi i\,t}\,\Disc\,\mathcal G_p\Bigl(\frac{1}{t}\Bigr),
		\qquad
		0<t<1/\zeta_c^2.
	\end{equation}
\end{theorem}

\begin{proof}
	Proposition~\ref{prop:stieltjes} constructs a representing measure
	\(\nu_p\) that is absolutely continuous on the open interval
	\((0,1/\zeta_c^2)\), with density given by the jump formula
	\eqref{eq:mu-jump}. In the range \(1\le p\le s\),
	Proposition~\ref{prop:mu-positive} shows that the same function
	\(\mathcal G_p\) also admits a positive representing measure. By uniqueness
	of the Stieltjes transform of a finite compactly supported measure, this
	positive measure coincides with the \(\nu_p\) from
	Proposition~\ref{prop:stieltjes}. Hence \(\nu_p\) is positive and its
	density in the spectral variable \(t\) can therefore be written in the Stieltjes--Perron form displayed in
	\eqref{eq:perron}; see, for example, \cite[Ch.~III]{Akhiezer1965}.
\end{proof}

\begin{figure}[t]
	\centering
	\includegraphics[width=\textwidth]{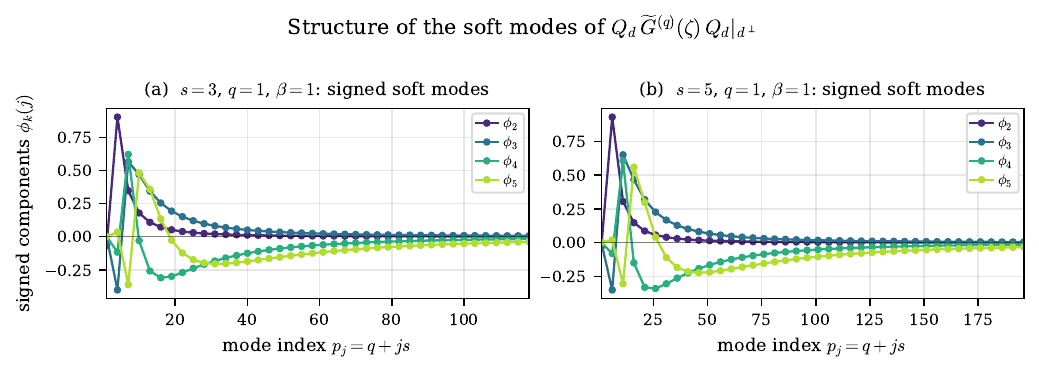}
	\caption{
	Structure of the first soft modes of
	\(\widetilde C^{(q)}_{*,\perp}\) for \(s=3,5\), with \(q=1\),
	\(\beta=1\), and \(N=40\).
	The curves show the signed components of the eigenvectors
	\(\phi_2^{(q)},\phi_3^{(q)},\phi_4^{(q)},\phi_5^{(q)}\),
	corresponding to the soft eigenvalues
	\(\mu_{2,*}^{(q)},\mu_{3,*}^{(q)},\mu_{4,*}^{(q)},\mu_{5,*}^{(q)}\).
	As the mode index increases, the eigenvectors become more oscillatory
	while remaining concentrated near low values of the lattice index
	\(p_j=q+js\).
}
	\label{fig:soft-structure}
\end{figure}

\begin{remark}[Nodal structure of the soft eigenvectors]
	\label{rem:nodal}
	The Jacobi realization suggests a further structural property of the
	compact remainder $\widetilde C^{(q)}_*$ introduced in
	Theorem~\ref{thm:rankone-decomp}.  Let
	$\phi_2^{(q)},\phi_3^{(q)},\ldots$ denote the eigenvectors of
	$\widetilde C^{(q)}_{*,\perp}$, ordered by decreasing eigenvalue
	$\mu_{2,*}^{(q)}\ge\mu_{3,*}^{(q)}\ge\cdots$, and view each $\phi_k^{(q)}$
	as a function on the lattice $\{p_j=q+js\}_{j\ge0}$.  Numerically,
	$\phi_k^{(q)}$ exhibits exactly $k-1$ sign changes in the index~$j$
	(Figure~\ref{fig:soft-structure}), a pattern stable across
	all values of $s$, $q$, and $\beta$ tested.  This oscillation count
	is the hallmark of classical Sturm oscillation for self-adjoint Jacobi
	matrices \cite[Ch.~III]{Akhiezer1965}.  For the range $1\le p\le s$,
	each scalar factor $\mathcal G_{p_j}$ admits such a Jacobi realization by
	Theorem~\ref{thm:weyl}. The observed nodal property of the full
	matrix $\widetilde C^{(q)}_*$ is therefore consistent with the scalar Jacobi framework. It is not, however, a direct corollary of that framework, since the Gram entries combine contributions from different Jacobi operators
	$J_{p_j}$.  We leave a rigorous nodal theorem for
	$\widetilde C^{(q)}_*$ as an open problem.
\end{remark}

The results of this section provide the analytic framework for the
intermediate regime \(\zeta_c<\zeta<\zeta_{\mathrm{univ}}\). The generating
functions \(\mathcal G_p(u)\) extend to the slit plane via their
hypergeometric structure, have a resonant singularity of parametric excess
\(\gamma=2\) at the branch point, and admit a nonvanishing edge density; for
\(1\le p\le s\), \(\mathcal G_p\) also has a Jacobi--Weyl realization.
Thus the scalar spectral quantities persist beyond the loss of the weighted compact
operator picture. Geometric singularity on the domain boundary occurs only at
\(\zeta=\zeta_{\mathrm{univ}}\): for \(s\ge3\) this is cusp formation,
whereas for \(s=2\) the critical map degenerates to the classical Joukowski
segment. We treat this geometric threshold next in
Section~\ref{sec:univalence}.

\FloatBarrier


\section{The geometric threshold \texorpdfstring{$\zeta_{\mathrm{univ}}$}{zeta-univ}}
\label{sec:univalence}


We now return to the geometric side of the problem. For the polynomial map
\[
	f(w)=rw+aw^{1-s},
	\qquad
	\zeta:=\frac{a}{r}>0,
\]
the parameter
\[
	\zeta_{\mathrm{univ}}=\frac{1}{s-1}
\]
marks the loss of univalence of the exterior conformal map. This threshold
is independent of the weighted spectral theory developed in
Section~\ref{sec:regime-I} and of the scalar continuation theory developed in
Section~\ref{sec:regime-2}. The purpose of the present section is to compare
these two notions of criticality.

We first characterize the geometric transition itself. Below
\(\zeta_{\mathrm{univ}}\) the boundary is a smooth Jordan curve. At
\(\zeta_{\mathrm{univ}}\), the critical geometry depends on \(s\): for
\(s\ge3\) the boundary develops semicubical cusps, whereas for \(s=2\) the
Joukowski trace degenerates to a line segment. Above
\(\zeta_{\mathrm{univ}}\), the map is no longer injective. For \(s\ge3\) the
boundary self-intersects, while for \(s=2\) the unit-circle trace remains an
ellipse although the exterior map ceases to be one-to-one. We then show that
\[
	\zeta_c<\zeta_{\mathrm{univ}},
\]
so the logarithmic spectral instability from
Theorems~\ref{thm:rankone-decomp} and
\ref{thm:spectral-asymptotics} occurs while the conformal map is still
univalent. Finally, we prove that the analytically continued scalar Gram quantities
remain finite at \(u=\zeta_{\mathrm{univ}}^2\). Thus the spectral
singularity at \(\zeta_c\) and the geometric singularity at
\(\zeta_{\mathrm{univ}}\) are genuinely distinct.

\subsection{Geometric threshold}
\label{subsec:univ-geom}

For the comparison with the analytic threshold \(\zeta_c\), we only need the
location of the geometric threshold at which the exterior map ceases to be
univalent. The finer local description of the boundary singularity at the
threshold is classical.

\begin{proposition}[Geometric threshold]
	\label{prop:univ-threshold}
	Let
	\[
		f(w)=rw+aw^{1-s}=r\bigl(w+\zeta w^{1-s}\bigr),
		\qquad s\ge2,\qquad \zeta:=a/r>0.
	\]
	Then the univalence threshold on the exterior disk is
	\[
		\zeta_{\mathrm{univ}}=\frac{1}{s-1}.
	\]
	More precisely:
	\begin{enumerate}
		\item if \(\zeta<\zeta_{\mathrm{univ}}\), then \(f\) is univalent on
		      \(\{|w|>1\}\) and extends injectively to \(|w|=1\);
		\item if \(\zeta>\zeta_{\mathrm{univ}}\), then \(f\) is not univalent on
		      \(\{|w|>1\}\);
		\item if \(\zeta=\zeta_{\mathrm{univ}}\), then
		      \[
			      f'(w)=r\bigl(1-(s-1)\zeta w^{-s}\bigr)
		      \]
		      vanishes exactly at the \(s\) points \(w^s=1\) on the unit circle.
	\end{enumerate}
\end{proposition}

\begin{proof}
	For \(w=e^{i\theta}\), write
	\[
		z(\theta):=f(e^{i\theta})
		=
		r\bigl(e^{i\theta}+\zeta e^{-i(s-1)\theta}\bigr).
	\]
	Suppose \(z(\theta)=z(\phi)\). Then $
		e^{i\theta}-e^{i\phi}
		=
		\zeta(e^{-i(s-1)\phi}-e^{-i(s-1)\theta})$.
	Taking absolute values and setting \(\delta=(\theta-\phi)/2\), we obtain
	\[
		|\sin\delta|
		=
		\zeta\,|\sin((s-1)\delta)|
		\le
		(s-1)\zeta\,|\sin\delta|.
	\]
	Hence, if \((s-1)\zeta<1\), then \(\sin\delta=0\), so
	\(\theta\equiv\phi\pmod{2\pi}\). Thus the boundary trace is injective.
	Moreover, $
		f'(w)=r\bigl(1-(s-1)\zeta w^{-s}\bigr)$,
	so for \((s-1)\zeta<1\) all critical points satisfy
	\[
		|w|^s=(s-1)\zeta<1.
	\]
	Therefore \(f'(w)\neq0\) on \(|w|\ge1\), so \(f\) is a holomorphic local
	biholomorphism on \(\{|w|>1\}\). Since \(f(w)=rw+O(1)\) as \(w\to\infty\),
	the map is proper as a map from the exterior disk to its image. Hence \(f\)
	is a covering map of finite degree onto its image. The normalization at
	infinity shows that this degree is \(1\): for \(|z|\) sufficiently large,
	the equation \(f(w)=z\) has exactly one solution in \(|w|>1\), namely the
	branch with \(w=z/r+O(1)\). Therefore the covering is trivial, so \(f\) is
	univalent on \(\{|w|>1\}\) and extends injectively to \(|w|=1\).

	If \((s-1)\zeta>1\), then the equation \(f'(w)=0\) has solutions with
	\[
		|w|^s=(s-1)\zeta>1,
	\]
	so \(f'\) vanishes inside the exterior domain. Since a holomorphic
	injective map cannot have a critical point, \(f\) is not univalent on
	\(\{|w|>1\}\).

	Finally, if \((s-1)\zeta=1\), then \(f'(w)=0\) is equivalent to \(w^s=1\),
	so the critical points lie exactly on the unit circle.
\end{proof}

\begin{remark}
	At the critical value \(\zeta=\zeta_{\mathrm{univ}}\), the boundary
	singularity is standard. For \(s\ge3\), the image of \(|w|=1\) develops
	\(s\) semicubical cusps, while for \(s=2\), the Joukowski trace degenerates to a
	line segment. These local descriptions are classical and are not needed for
	the spectral analysis below.
\end{remark}

\subsection{Separation of the analytic and geometric thresholds}
\label{subsec:separation}

With the geometric threshold identified, we now compare it with the analytic
threshold $\zeta_c = (s-1)^{s-1}/s^s$.

\begin{proposition}[Separation of thresholds]
	\label{prop:separation}
	For every \(s\ge2\),
	\[
		\zeta_c<\zeta_{\mathrm{univ}}.
	\]
	Equivalently,
	\begin{equation}\label{eq:threshold-ratio}
		\frac{\zeta_c}{\zeta_{\mathrm{univ}}}
		=
		\Bigl(\frac{s-1}{s}\Bigr)^s
		<1.
	\end{equation}
	In particular, the logarithmic spectral transition described by
	Theorems~\ref{thm:rankone-decomp} and
	\ref{thm:spectral-asymptotics} takes place while the conformal map is still
	univalent and the boundary is still smooth.
\end{proposition}

\begin{proof}
	The identity \eqref{eq:threshold-ratio} is immediate:
	\[
		\frac{\zeta_c}{\zeta_{\mathrm{univ}}}
		=
		(s-1)\frac{(s-1)^{s-1}}{s^s}
		=
		\Bigl(\frac{s-1}{s}\Bigr)^s.
	\]
	Since \((s-1)/s<1\), the right-hand side is strictly less than \(1\).
\end{proof}

\subsection{\texorpdfstring{Regularity of the continued scalar quantities at $\zeta_{\mathrm{univ}}$}{Regularity of the continued scalar data at zeta-univ}}
\label{subsec:univ-spectral}

We finally show that the continued scalar quantities do not develop any new
singularity at the geometric threshold. This is the final step in separating
analytic and geometric criticality. Since \(\zeta_{\mathrm{univ}}>\zeta_c\), the point
\[
	u_{\mathrm{univ}}:=\zeta_{\mathrm{univ}}^2
\]
lies on the branch cut \([\zeta_c^2,\infty)\) of the analytically continued
functions \(\mathcal G_p(u)\). However, \(u_{\mathrm{univ}}\) is an interior
point of the cut, not its endpoint. The only singular endpoint produced by
the hypergeometric continuation is \(u=\zeta_c^2\). Thus one expects finite
lateral values at \(u=u_{\mathrm{univ}}\), and the next proposition confirms
this. The key point is that every interior point of the cut is an ordinary
point of the hypergeometric differential equation.

\begin{lemma}[Interior points of the cut are ordinary points]
	\label{lem:interior-cut-ordinary}
	Fix \(p\ge1\) and \(u_0>\zeta_c^2\). Let \(\mathcal G_{p,+}\) and
	\(\mathcal G_{p,-}\) denote the analytic continuations of the germ of
	\(\mathcal G_p\) from \(u=0\) to a neighborhood of \(u_0\) through the upper
	and lower half-planes, respectively. Then each branch extends holomorphically
	to a full disk centered at \(u_0\). In particular,
	\[
		\mathcal G_{p,\pm}(u_0),\qquad
		\mathcal G'_{p,\pm}(u_0),\qquad
		\mathcal G''_{p,\pm}(u_0)
	\]
	are finite.
\end{lemma}

\begin{proof}
	By Proposition~\ref{prop:Gp-hypergeom}, \(\mathcal G_p\) satisfies the
	generalized hypergeometric differential equation in the variable
	\(\xi=\zeta_c^{-2}u\). This equation is Fuchsian, and its only finite
	singular points are \(\xi=0\) and \(\xi=1\). Equivalently, in the variable
	\(u\) the only finite singular points are \(u=0\) and \(u=\zeta_c^2\); see
	\cite[\S16.8,\S16.11]{NIST:DLMF}. Hence every \(u_0>\zeta_c^2\) with
	\(u_0\neq \zeta_c^2\) is an ordinary point of the differential equation.

	Choose \(r>0\) so small that the closed disk \(\overline{D(u_0,r)}\) avoids
	\(\{0,\zeta_c^2\}\). On \(D(u_0,r)\), the differential equation has analytic
	coefficients. Standard local ODE theory therefore implies that any solution
	defined on a connected open subset of \(D(u_0,r)\) extends uniquely to a
	holomorphic solution on the whole disk. Applying this to the branches
	\(\mathcal G_{p,+}\) and \(\mathcal G_{p,-}\), initially defined on the
	upper and lower half-disks, gives the claimed holomorphic extensions. Their
	values and first two derivatives at \(u_0\) are therefore finite.
\end{proof}

\begin{proposition}[Regularity at the geometric threshold]
	\label{prop:univ-regularity}
	For every \(p\ge1\), the lateral values
	\[
		\mathcal G_p(u\pm i0),
		\qquad
		\sigma_p^{\mathrm{cont}}(u\pm i0),
	\]
	are finite for every \(u>\zeta_c^2\). In particular, they are finite at $u=\zeta_{\mathrm{univ}}^2$,
	and the discontinuity density satisfies
	\[
		\rho_p(\zeta_{\mathrm{univ}}^2)<\infty.
	\]
\end{proposition}

\begin{proof}
	Fix \(u_0>\zeta_c^2\). By Lemma~\ref{lem:interior-cut-ordinary}, each lateral
	branch of \(\mathcal G_p\) extends holomorphically to a neighborhood of
	\(u_0\). Hence the limits
	\[
		\mathcal G_p(u_0\pm i0)
	\]
	exist and are finite, and so do the first two lateral derivatives.

	Now \(\sigma_p^{\mathrm{cont}}\) is obtained from \(\mathcal G_p\) by the Euler
	operator
	\[
		\sigma_p^{\mathrm{cont}}(u)
		=
		\frac1p\Bigl(p+s\,u\frac{d}{du}\Bigr)^2\mathcal G_p(u),
	\]
	whose coefficients are analytic at every finite \(u\). Therefore
	\(\sigma_p^{\mathrm{cont}}(u_0\pm i0)\) is finite for every \(u_0>\zeta_c^2\),
	and in particular at \(u_0=\zeta_{\mathrm{univ}}^2\).

	Finally, Definition~\ref{def:discontinuity} gives
	\[
		\rho_p(u_0)
		=
		-\frac{1}{2\pi i}
		\bigl(
		\sigma_p^{\mathrm{cont}}(u_0+i0)-\sigma_p^{\mathrm{cont}}(u_0-i0)
		\bigr),
	\]
	which is therefore finite as well.\end{proof}

The three regimes of the problem may now be summarized as follows: (i) For
\(0<\zeta<\zeta_c\), each symmetry sector is described by a compact weighted
Gram operator with one stiff eigenvalue diverging logarithmically as
\(\zeta\uparrow\zeta_c\), while the remaining sectorial spectrum stays
bounded. Equivalently, there are at most \(s\) logarithmically diverging
eigenvalues globally. (ii) For \(\zeta_c<\zeta<\zeta_{\mathrm{univ}}\), the
weighted operator realization breaks down, but the scalar quantities
\(\mathcal G_p\), \(\sigma_p^{\mathrm{cont}}\), and \(\rho_p\) admit a
canonical analytic continuation, with Jacobi--Weyl interpretation for
\(1\le p\le s\). (iii) At \(\zeta=\zeta_{\mathrm{univ}}\), the conformal map
loses univalence. For \(s\ge3\) the boundary develops cusp singularities,
whereas for \(s=2\) the critical map degenerates to the Joukowski segment. In
both cases, the continued scalar quantities remain finite.

This completes the proof that analytic criticality and geometric criticality
are distinct in the symmetric one-harmonic family.


\section*{Conclusion}


We have shown that for the \(s\)-fold symmetric one-harmonic polynomial
family the first spectral instability in the weighted block realizations of
the mixed Hessian of the dispersionless Toda \(\tau\)-function is governed by
analytic criticality of the inverse map, not by the later geometric loss of
univalence of the conformal map. The square-root singularity of the inverse map produces the
borderline coefficient decay that drives the transition.

On the subcritical side, for any fixed \(\beta>0\), the weighted Hilbert
realization of Definition~\ref{def:weighted-conj} separates this mechanism
into one stiff direction and a bounded soft sector in each symmetry block.
Accordingly, in each such weighted realization each sector contains exactly
one logarithmically diverging eigenvalue, while the remaining sectorial
spectrum stays bounded. After removing the spike direction, each fixed soft
eigenvalue converges to the corresponding eigenvalue of a compact limiting
remainder. At coefficient level, the same phenomenon is encoded by a single
dominant singular direction in each sector, from which the weighted spikes
are realized. Beyond the analytic threshold, the weighted compact
realization ceases to apply, but the scalar Gram quantities remain well defined
after analytic continuation. They are described by generalized hypergeometric
functions, admit a Cauchy--Stieltjes representation, and in the positive
range fit naturally into the Jacobi/Weyl framework.

The geometric threshold lies strictly beyond the analytic one. Consequently,
the weighted block theory becomes spectrally singular while the conformal map
is still univalent and the boundary is still smooth, and the continued scalar
quantities have finite upper and lower lateral boundary values at the univalence
threshold. In this sense, analytic and geometric criticality are genuinely
distinct in the present family.

The result also has a direct interpretation in adjacent frameworks. In the
Laplacian-growth language, the rank-one logarithmic spike singles out a
distinguished \emph{stiff} deformation mode before any geometric singularity
appears. In the matrix-model or potential-theoretic language, the same
statement says that the mixed second-order response concentrates onto a single
dominant sectorial direction at criticality. These interpretations lie outside the proofs. They serve only to place the theorem in a broader geometric and physical context.

The argument isolates a structural mechanism that should persist beyond the
one-harmonic leaf: a dominant branch point, borderline coefficient decay, and
positive Gram structure. For more general polynomial conformal maps, one
expects several competing dominant orbits and hence a finite-rank version of
the present instability mechanism. Another natural problem is to understand
the bounded soft sector more conceptually, possibly through hidden Jacobi or
total-positivity structure, and to sharpen the leading-order theory by
deriving asymptotics for spectral gaps, sector dependence, and soft
eigenvalues.


\section*{Acknowledgments}


This work is an output of the research project
``Symmetry. Information. Chaos''
implemented as a part of the Basic Research Program at
National Research University Higher School
of Economics (HSE University).

\appendix

\section{Asymptotics of Raney numbers}
\label{app:raney-asymp}

This appendix presents the Stirling-type asymptotics for the Raney numbers used
in the Gram estimates of Section~\ref{sec:regime-I} and in the branch-point
coefficient computation from Appendix~\ref{app:B1}.

\begin{proposition}[Explicit $m^{-3/2}$ asymptotics]\label{prop:raney-asymp}
	Fix integers $s\ge2$, $p\ge1$, and set
	\[
		\zeta_c:=\frac{(s-1)^{s-1}}{s^s},\qquad M:=\frac{s}{s-1}.
	\]
	As $m\to\infty$,
	\begin{equation}\label{eq:app-Raney-asymp}
		R_{s,p}(m)
		=
		A_{s,p}\,\zeta_c^{-m}\,m^{-3/2}\bigl(1+O(m^{-1})\bigr),
		\qquad
		A_{s,p}
		=
		\frac{p}{\sqrt{2\pi}}\,
		\frac{s^{\,p-\frac12}}{(s-1)^{\,p+\frac12}}
		=
		\frac{p\,M^{p}}{\sqrt{2\pi s(s-1)}}.
	\end{equation}
	Consequently,
	\begin{equation}\label{eq:app-Raney-square-asymp}
		R_{s,p}(m)^2
		=
		\frac{p^{2}}{2\pi}\,
		\frac{s^{\,2p-1}}{(s-1)^{\,2p+1}}\,
		\zeta_c^{-2m}\,m^{-3}\bigl(1+O(m^{-1})\bigr).
	\end{equation}
\end{proposition}

\begin{proof}
	Starting from \eqref{eq:raney-closed} and writing factorials via Gamma functions,
	\[
		R_{s,p}(m)
		=
		\frac{p}{sm+p}\,
		\frac{\Gamma(sm+p+1)}{\Gamma(m+1)\,\Gamma((s-1)m+p+1)}
		=
		\frac{p\,\Gamma(sm+p)}{\Gamma(m+1)\,\Gamma((s-1)m+p+1)}.
	\]
	For fixed \(s,p\) we apply Stirling's formula
	\[
		\Gamma(z)=\sqrt{2\pi}\,z^{z-\frac12}e^{-z}\bigl(1+O(z^{-1})\bigr),
		\qquad z\to+\infty,
	\]
	to the three Gamma factors with arguments \(sm+p\), \(m+1\), and
	\((s-1)m+p+1\).  The exponential terms \(e^{-sm}e^{m}e^{(s-1)m}\) cancel,
	while the power terms produce
	\[
		\frac{(sm)^{sm+p-\frac12}}{m^{m+\frac12}((s-1)m)^{(s-1)m+p+\frac12}}
		=
		\Bigl(\frac{s^s}{(s-1)^{s-1}}\Bigr)^m
		m^{-3/2}\frac{s^{p-\frac12}}{(s-1)^{p+\frac12}}
		\bigl(1+O(m^{-1})\bigr).
	\]
	Since \(\zeta_c^{-1}=s^s/(s-1)^{s-1}\), this gives
	\eqref{eq:app-Raney-asymp}. Squaring yields
	\eqref{eq:app-Raney-square-asymp}.
\end{proof}

\begin{proposition}[Uniform upper bound]\label{prop:raney-uniform}
	There exists a constant $C_s>0$ such that for all integers $p\ge1$ and $m\ge1$,
	\begin{equation}\label{eq:raney-uniform}
		R_{s,p}(m)\le C_s\,p\,M^{p}\,\zeta_c^{-m}\,m^{-3/2}.
	\end{equation}
	Moreover, the same bound holds with $m^{-3/2}$ replaced by $(1+m)^{-3/2}$ for all $m\ge0$.
\end{proposition}

\begin{proof}
	We start from the closed form \eqref{eq:raney-closed}:
	\[
		R_{s,p}(m)=\frac{p}{sm+p}\binom{sm+p}{m},\qquad m\ge1.
	\]
	Set $N:=sm+p$ and $K:=N-m=(s-1)m+p$. Using a standard two-sided Stirling bound
	(e.g.\ Robbins' bound), there is an absolute constant $C>0$ such that
	\[
		\binom{N}{m}
		\le
		C\,
		\frac{N^{N+\frac12}}{m^{m+\frac12}\,K^{K+\frac12}}
		=
		C\,
		\frac{1}{\sqrt{m}}\,
		\sqrt{\frac{N}{K}}\,
		\exp\!\bigl(N\log N-m\log m-K\log K\bigr).
	\]
	Write $t:=p/m\ge0$, so that $N=m(s+t)$ and $K=m(s-1+t)$. The exponent simplifies to
	\[
		N\log N-m\log m-K\log K
		=
		m\Bigl[(s+t)\log(s+t)-(s-1+t)\log(s-1+t)\Bigr].
	\]
	Define
	\[
		\Phi(t):=\log\Bigl(\frac{s^s}{(s-1)^{s-1}}\Bigr)+t\log\Bigl(\frac{s}{s-1}\Bigr)
		-\Bigl[(s+t)\log(s+t)-(s-1+t)\log(s-1+t)\Bigr].
	\]
	A direct computation gives $\Phi(0)=\Phi'(0)=0$ and
	\[
		\Phi''(t)=\frac{1}{(s+t)(s-1+t)}>0,
	\]
	so $\Phi$ is convex with a global minimum at $t=0$. Hence $\Phi(t)\ge0$ for all $t\ge0$, i.e.
	\[
		(s+t)\log(s+t)-(s-1+t)\log(s-1+t)
		\le
		\log(\zeta_c^{-1})+t\log(M).
	\]
	Substituting back yields
	\[
		\exp\bigl(N\log N-m\log m-K\log K\bigr)\le \zeta_c^{-m}\,M^{p}.
	\]
	Also $\sqrt{N/K}\le \sqrt{s/(s-1)}=\sqrt{M}$ and $\frac{p}{sm+p}\le \frac{p}{sm}$.
	Combining these bounds gives
	\[
		R_{s,p}(m)
		\le
		\frac{p}{sm}\cdot C\,\frac{1}{\sqrt{m}}\cdot \sqrt{M}\cdot \zeta_c^{-m}M^p
		\le
		C_s\,p\,M^{p}\,\zeta_c^{-m}\,m^{-3/2},
	\]
	for a constant $C_s$ depending only on $s$. The extension to $m=0$ follows by enlarging $C_s$,
	since $R_{s,p}(0)=1$ and $(1+m)^{-3/2}=1$ at $m=0$.
\end{proof}

\begin{lemma}[Uniform one-term expansion in the tail region]\label{lem:raney-uniform-expansion}
	Fix \(s\ge2\) and \(0<\theta<1\).  There exists a constant \(C_{s,\theta}>0\)
	such that for all integers \(p\ge1\) and \(m\ge1\) satisfying \(p\le \theta m\),
	\begin{equation}\label{eq:raney-uniform-expansion}
		R_{s,p}(m)
		=
		A_{s,p}\,\zeta_c^{-m}\,m^{-3/2}\Bigl(1+\varepsilon_{s,p}(m)\Bigr),
		\qquad
		|\varepsilon_{s,p}(m)|\le \frac{C_{s,\theta}}{m},
	\end{equation}
	where $A_{s,p}$ is as in \eqref{eq:app-Raney-asymp}.
\end{lemma}

\begin{proof}
	Write \(t:=p/m\), so \(0\le t\le\theta\), and write \(p=tm\).
	Using the Gamma-form expression from the proof of
	Proposition~\ref{prop:raney-asymp}, consider
	\[
		\log\!\left(
		\frac{R_{s,p}(m)}{A_{s,p}\,\zeta_c^{-m}\,m^{-3/2}}
		\right).
	\]
	Applying the logarithmic Stirling expansion
	\[
		\log\Gamma(z)=\Bigl(z-\frac12\Bigr)\log z-z+\frac12\log(2\pi)+O(z^{-1})
	\]
	to the three Gamma factors with arguments \(m(s+t)\), \(m+1\), and
	\(m(s-1+t)+1\), one obtains
	\[
		\log R_{s,p}(m)
		=
		-m\log\zeta_c+\log A_{s,p}-\frac32\log m+O(m^{-1}),
	\]
	uniformly for \(t\in[0,\theta]\). Equivalently,
	\[
		\log\!\left(
		\frac{R_{s,p}(m)}{A_{s,p}\,\zeta_c^{-m}\,m^{-3/2}}
		\right)
		=
		O(m^{-1}),
	\]
	uniformly for \(t\in[0,\theta]\). Here the terms of order \(m\), the
	\(\log m\) contribution, and the explicit \(t\)-dependent prefactor
	encoded in \(A_{s,p}\) cancel identically. The remainder is uniform because
	all auxiliary functions of \(t\) are smooth on the compact interval
	\([0,\theta]\). Exponentiating yields
	\[
		\frac{R_{s,p}(m)}{A_{s,p}\,\zeta_c^{-m}\,m^{-3/2}}
		=1+O(m^{-1}),
	\]
	uniformly for \(p\le \theta m\), which is exactly
	\eqref{eq:raney-uniform-expansion}.
\end{proof}

\section{Rank--one tail extraction in the Gram blocks}
\label{app:rankone-form}

This appendix supplies the coefficient-level tail extraction used in the proof
of Theorem~\ref{thm:rankone-decomp}. Throughout we fix a symmetry sector
\(q\in\{1,\dots,s\}\), write
\[
    p_j:=q+js,\qquad j\in\NN_0,
\]
and use the Gram vectors \(\bm v^{(p)}\) from \eqref{eq:gram-vector}. Recall the
block matrix \(G^{(q)}(\zeta)\) determined coefficientwise by the Gram
representation in \eqref{eq:Hq-VV}, and its weighted conjugate
\[
    \widetilde G^{(q)}(\zeta)=\mathcal W^{-1}G^{(q)}(\zeta)\mathcal W^{-1},
\]
where \(\mathcal W=\mathrm{diag}(w_j)\) and \(w_j\) is given by
\eqref{eq:weights}.

\subsection{Entry formula}

\begin{lemma}[Entries of the Gram block]\label{lem:Gram-entry}
	For $j_1\le j_2$ (and $\Delta:=j_2-j_1\ge0$) one has
	\begin{equation}\label{eq:app-Gram-entry}
		G^{(q)}_{j_1j_2}(\zeta)
		=
		\sum_{m=0}^{\infty}
		\frac{(p_{j_2}+sm)^{2}}{\sqrt{p_{j_1}p_{j_2}}}\,
		R_{s,p_{j_1}}(m+\Delta)\,R_{s,p_{j_2}}(m)\,\zeta^{\,2m+\Delta}.
	\end{equation}
	Consequently,
	\begin{equation}\label{eq:app-Gtilde-entry-clean}
		\widetilde G^{(q)}_{j_1j_2}(\zeta)
		=
		\frac{1}{w_{j_1}w_{j_2}}\,G^{(q)}_{j_1j_2}(\zeta),
		\qquad j_1,j_2\in\NN_0.
	\end{equation}
\end{lemma}

\begin{proof}
	By definition,
	$G^{(q)}_{j_1j_2}(\zeta)=\langle \bm v^{(p_{j_1})},\bm v^{(p_{j_2})}\rangle$
	in $\ell^2(\NN)$, and $\bm v^{(p)}$ is supported on indices $p+sm$.
	For $j_1\le j_2$, the supports overlap precisely at
	\(
	p_{j_2}+sm=p_{j_1}+s(m+\Delta),\ m\ge0.
	\)
	Inserting \eqref{eq:gram-vector} at these indices yields \eqref{eq:app-Gram-entry}.
	The conjugated entry relation \eqref{eq:app-Gtilde-entry-clean} is immediate from
	\eqref{eq:Gtilde}.
\end{proof}

\subsection{The logarithmic tail and the singular vector}

Write
\[
	\eta:=\frac{\zeta}{\zeta_c}\in(0,1),
	\qquad
	\delta:=1-\eta^2=1-\frac{\zeta^2}{\zeta_c^2},
	\qquad
	L(\zeta)=\sum_{m\ge1}\frac{\eta^{2m}}{m}=\log\frac{1}{1-\zeta^2/\zeta_c^2}.
\]
By Remark~\ref{rem:coeff-asymp} (or Appendix~\ref{app:raney-asymp}), for each
fixed $p\ge1$,
\begin{equation}\label{eq:app-Raney-asymp-used}
	R_{s,p}(m)
	=
	A_{s,p}\,\zeta_c^{-m}\,m^{-3/2}\bigl(1+O(m^{-1})\bigr),
	\qquad m\to\infty,
\end{equation}
with an explicit constant $A_{s,p}>0$.  Define the intrinsic amplitudes and their weighted realizations by
\begin{equation}\label{eq:app-d-def}
	d^{(q)}_j:=\frac{s\,A_{s,p_j}}{\sqrt{p_j}},\qquad \widetilde d^{(q)}_j:=\frac{d^{(q)}_j}{w_j},
	\qquad j\in\NN_0.
\end{equation}
Then $\widetilde{\bm d}^{(q)}:=(\widetilde d^{(q)}_j)_{j\ge0}\in\ell^2(\NN_0)$,
and it is the vector appearing in Theorem~\ref{thm:rankone-decomp}.

\begin{lemma}[Uniform exponential gap away from the critical slope]\label{lem:raney-gap}
	Fix $s\ge2$ and $\theta>0$.  There exist constants $C_{s,\theta},c_{s,\theta}>0$
	such that for all integers $p,m\ge1$ satisfying $p\ge\theta m$,
	\[
		R_{s,p}(m)
		\le
		C_{s,\theta}\,p\,M^{p}\,\zeta_c^{-m}\,m^{-3/2}\,e^{-c_{s,\theta}p}.
	\]
\end{lemma}

\begin{proof}
	Write $t:=p/m\in[\theta,\infty)$.  By the Stirling estimate used in the proof
	of Proposition~\ref{prop:raney-uniform}, one has
	\[
		R_{s,p}(m)
		\le
		C_s\,p\,m^{-3/2}\exp\bigl(m\Psi_s(t)\bigr),
	\]
	where
	\[
		\Psi_s(t):=(s+t)\log(s+t)-(s-1+t)\log(s-1+t).
	\]
	Equivalently,
	\[
		R_{s,p}(m)
		\le
		C_s\,p\,M^p\,\zeta_c^{-m}\,m^{-3/2}
		\exp\Bigl(p\,\Xi_s(t)\Bigr),
	\]
	with
	\[
		\Xi_s(t):=\frac{\Psi_s(t)-\log(\zeta_c^{-1})}{t}-\log M.
	\]
	The function $\Xi_s$ is continuous on $[\theta,\infty)$.  The convexity
	argument used in Proposition~\ref{prop:raney-uniform} yields
	\[
		\Psi_s(t)-t\log M\le \log(\zeta_c^{-1}),
	\]
	with equality only at $t=0$, so $\Xi_s(t)<0$ for every $t>0$.
	Moreover, $\Xi_s(t)\to-\log M<0$ as $t\to\infty$.  Hence
	\[
		\sup_{t\ge\theta}\Xi_s(t)=-c_{s,\theta}<0,
	\]
	which gives the stated bound.
\end{proof}

\begin{lemma}[Logarithmic tail with Hilbert--Schmidt remainder]\label{lem:entrywise-log}
	Let $\eta=\zeta/\zeta_c\in(0,1)$ and $L(\zeta)=\log\frac{1}{1 - \eta^2}$.
	Define the rank--one Toeplitz kernel
	\[
		(K_\eta)_{j_1j_2}:=\eta^{|j_1-j_2|}\,\widetilde d^{(q)}_{j_1}\,\widetilde d^{(q)}_{j_2},
		\qquad j_1,j_2\in\NN_0.
	\]
	Then, for every $0<\zeta<\zeta_c$, the operator
	\[
	R^{(q)}(\zeta):=\widetilde G^{(q)}(\zeta)-L(\zeta)\,K_\eta
	\]
	is Hilbert--Schmidt on $\ell^2(\NN_0)$, with $\sup_{\zeta<\zeta_c}\|R^{(q)}(\zeta)\|_{\mathrm{HS}}<\infty$.
 Moreover, as $\zeta\uparrow\zeta_c$, $R^{(q)}(\zeta)$ converges in Hilbert--Schmidt norm
 to a limit Hilbert--Schmidt operator $R_*^{(q)}$, hence also in operator norm.

	In particular, for each fixed $(j_1,j_2)$ one has the entrywise expansion
	\begin{equation}\label{eq:app-entrywise-log}
		\widetilde G^{(q)}_{j_1j_2}(\zeta)
		=
		\widetilde d^{(q)}_{j_1}\,\widetilde d^{(q)}_{j_2}\,\eta^{|j_1-j_2|}\,L(\zeta)
		+(R_*^{(q)})_{j_1j_2}
		+o(1),
		\qquad \zeta\uparrow\zeta_c,
	\end{equation}
	and the error is $o(1)$ in $\ell^2$--matrix (Hilbert--Schmidt) sense.
\end{lemma}

\begin{proof}
	By symmetry of the Gram kernel, it is enough to treat the case \(j_1\le j_2\). Write \(\Delta:=j_2-j_1\ge0\).  By
	Lemma~\ref{lem:Gram-entry},
	\[
		\widetilde G^{(q)}_{j_1j_2}(\zeta)=\sum_{m\ge0}U_m(j_1,j_2;\zeta),
	\]
	where
	\[
		U_{m}(j_1,j_2;\zeta)
		:=
		\frac{1}{w_{j_1}w_{j_2}}\,
		\frac{(p_{j_2}+sm)^{2}}{\sqrt{p_{j_1}p_{j_2}}}\,
		R_{s,p_{j_1}}(m+\Delta)\,R_{s,p_{j_2}}(m)\,\zeta^{\,2m+\Delta}.
	\]
	Since \(p_{j_2}\ge p_{j_1}\) and \(p_{j_2}\ge \Delta\), the cutoff
	\[
		M_{j_1j_2}:=\max\{2p_{j_1},\,2p_{j_2},\,2\Delta,\,1\}
	\]
	reduces here to \(M_{j_1j_2}=2p_{j_2}\).  We split the sum into the initial
	range \(0\le m<2p_{j_2}\) and the tail \(m\ge2p_{j_2}\).

	\smallskip
	\noindent\emph{Tail region.}
	If \(m\ge2p_{j_2}\), then \(p_{j_2}\le m/2\) and
	\(p_{j_1}\le (m+\Delta)/2\), so Lemma~\ref{lem:raney-uniform-expansion}
	applies to both Raney factors (with \(\theta=\frac12\)). More explicitly,
	\[
		R_{s,p_{j_1}}(m+\Delta)
		=
		A_{s,p_{j_1}}\zeta_c^{-m-\Delta}(m+\Delta)^{-3/2}
		\Bigl(1+O((m+1)^{-1})\Bigr),
	\]
	and
	\[
		R_{s,p_{j_2}}(m)
		=
		A_{s,p_{j_2}}\zeta_c^{-m}m^{-3/2}
		\Bigl(1+O((m+1)^{-1})\Bigr),
	\]
	uniformly for \(m\ge2p_{j_2}\). Since \(j_1\le j_2\), the quantities
	\(p_{j_1},p_{j_2},\Delta\) are fixed along the tail sum, while
	\[
		(p_{j_2}+sm)^2=s^2m^2\Bigl(1+O((m+1)^{-1})\Bigr)
	\]
	and
	\[
		(m+\Delta)^{-3/2}m^{-3/2}
		=
		m^{-3}\Bigl(1+O((m+1)^{-1})\Bigr).
	\]
	Hence
	\[
		\frac{(p_{j_2}+sm)^2}{\sqrt{p_{j_1}p_{j_2}}}\,
		R_{s,p_{j_1}}(m+\Delta)\,R_{s,p_{j_2}}(m)\,\zeta^{2m+\Delta}
		=
		\frac{s^2A_{s,p_{j_1}}A_{s,p_{j_2}}}{\sqrt{p_{j_1}p_{j_2}}}\,
		\eta^\Delta\frac{\eta^{2m}}{m}
		\Bigl(1+O((m+1)^{-1})\Bigr).
	\]
	Using \eqref{eq:app-d-def} and \(m^{-1}=(m+1)^{-1}+O((m+1)^{-2})\), we
	obtain
	\begin{equation}\label{eq:Sm-decomp}
		U_m(j_1,j_2;\zeta)
		=
		\widetilde d^{(q)}_{j_1}\widetilde d^{(q)}_{j_2}\,\eta^{\Delta}\,
		\frac{\eta^{2m}}{m+1}
		+E_m^{\mathrm{tail}}(j_1,j_2;\zeta),
	\end{equation}
	with
	\begin{equation}\label{eq:Sm-error}
		|E_m^{\mathrm{tail}}(j_1,j_2;\zeta)|
		\le
		C\,\widetilde d^{(q)}_{j_1}\widetilde d^{(q)}_{j_2}\,\eta^{\Delta}\,
		\frac{\eta^{2m}}{(m+1)^2},
		\qquad m\ge 2p_{j_2},
	\end{equation}
	where \(C\) depends only on \(s\) and \(\beta\).  Summing over \(m\ge2p_{j_2}\)
	shows that
	\[
		E^{\mathrm{tail}}_{j_1j_2}(\zeta)
		:=\sum_{m\ge2p_{j_2}}E_m^{\mathrm{tail}}(j_1,j_2;\zeta)
	\]
	defines a uniformly Hilbert--Schmidt kernel, since
	\[
		|E^{\mathrm{tail}}_{j_1j_2}(\zeta)|
		\le
		C\,\widetilde d^{(q)}_{j_1}\widetilde d^{(q)}_{j_2}
		\sum_{m\ge0}\frac{1}{(m+1)^2}
		\le
		C'\,\widetilde d^{(q)}_{j_1}\widetilde d^{(q)}_{j_2},
	\]
	and \(\widetilde{\bm d}^{(q)}\in\ell^2(\NN_0)\).

	\smallskip
	\noindent\emph{Initial range.}
	Write
	\[
		E^{\mathrm{init}}_{j_1j_2}(\zeta):=\sum_{0\le m<2p_{j_2}}U_m(j_1,j_2;\zeta).
	\]
	The term \(m=0\) is estimated by Proposition~\ref{prop:raney-uniform}:
	\[
		|U_0(j_1,j_2;\zeta)|
		\le
		C\,M^{-p_{j_2}}\,p_{j_1}^{-1-\beta}p_{j_2}^{1-\beta}(1+\Delta)^{-3/2}.
	\]
	For \(1\le m<2p_{j_2}\), Proposition~\ref{prop:raney-uniform} applies to
	\(R_{s,p_{j_1}}(m+\Delta)\), while Lemma~\ref{lem:raney-gap} with
	\(\theta=\frac12\) applies to \(R_{s,p_{j_2}}(m)\) because
	\(p_{j_2}\ge m/2\).  Using \(\zeta^{2m+\Delta}\le \zeta_c^{2m+\Delta}\), one gets
	\[
		|U_m(j_1,j_2;\zeta)|
		\le
		C\,e^{-c p_{j_2}}\,
		p_{j_1}^{-1-\beta}p_{j_2}^{-1-\beta}\,
		\frac{(p_{j_2}+sm)^2}{(m+1)^{3/2}(m+\Delta+1)^{3/2}}.
	\]
	Since \(m<2p_{j_2}\), we have \(p_{j_2}+sm\le (1+2s)p_{j_2}\), and therefore
	\[
		\sum_{m=1}^{2p_{j_2}-1}|U_m(j_1,j_2;\zeta)|
		\le
		C\,e^{-c p_{j_2}}\,
		p_{j_1}^{-1-\beta}p_{j_2}^{1-\beta}
		\sum_{m\ge1}\frac{1}{(m+1)^3}
		\le
		C'\,e^{-c p_{j_2}}\,p_{j_1}^{-1-\beta}p_{j_2}^{1-\beta}.
	\]
	Hence
	\[
		|E^{\mathrm{init}}_{j_1j_2}(\zeta)|
		\le
		C''\,e^{-c' p_{j_2}}\,p_{j_1}^{-1-\beta}p_{j_2}^{1-\beta},
	\]
	for constants \(C'',c'>0\) independent of \(\zeta\).  This bound is square
	summable in \((j_1,j_2)\), so \(E^{\mathrm{init}}(\zeta)\) is uniformly
	Hilbert--Schmidt.  Moreover, for each fixed \((j_1,j_2)\) the sum is finite and
	therefore converges termwise as \(\zeta\uparrow\zeta_c\).

	Summing \eqref{eq:Sm-decomp} over \(m\ge2p_{j_2}\) gives
	\[
		\sum_{m\ge2p_{j_2}}
		U_m(j_1,j_2;\zeta)
		=
		\widetilde d^{(q)}_{j_1}\widetilde d^{(q)}_{j_2}\,\eta^\Delta
		\sum_{m\ge2p_{j_2}}\frac{\eta^{2m}}{m+1}
		+
		E^{\mathrm{tail}}_{j_1j_2}(\zeta).
	\]
	Reinstating the omitted harmonic terms produces the correction
	\[
		E^{\mathrm{harm}}_{j_1j_2}(\zeta)
		:=
		\widetilde d^{(q)}_{j_1}\widetilde d^{(q)}_{j_2}\,\eta^\Delta
		\sum_{m=0}^{2p_{j_2}-1}\frac{\eta^{2m}}{m+1}.
	\]
	Since
	\[
		\sum_{m=0}^{2p_{j_2}-1}\frac{\eta^{2m}}{m+1}
		\le
		1+\log(1+2p_{j_2}),
	\]
	we obtain
	\[
		|E^{\mathrm{harm}}_{j_1j_2}(\zeta)|
		\le
		C\,\widetilde d^{(q)}_{j_1}\widetilde d^{(q)}_{j_2}\bigl(1+\log(1+p_{j_2})\bigr).
	\]
	Because \(\widetilde d^{(q)}_j\asymp p_j^{-1-\beta}\), the kernel on the
	right-hand side is Hilbert--Schmidt:
	\[
		\sum_{j_1,j_2\ge0}
		\widetilde d^{(q)}_{j_1}{}^2\widetilde d^{(q)}_{j_2}{}^2
		\bigl(1+\log(1+p_{j_2})\bigr)^2
		<\infty.
	\]

	Thus
	\[
		\widetilde G^{(q)}_{j_1j_2}(\zeta)
		=
		\widetilde d^{(q)}_{j_1}\widetilde d^{(q)}_{j_2}\,\eta^\Delta
		\sum_{m\ge0}\frac{\eta^{2m}}{m+1}
		+
		E^{\mathrm{init}}_{j_1j_2}(\zeta)
		+
		E^{\mathrm{tail}}_{j_1j_2}(\zeta)
		+
		E^{\mathrm{harm}}_{j_1j_2}(\zeta).
	\]
	Since \(\sum_{m\ge0}\frac{\eta^{2m}}{m+1}=\eta^{-2}L(\zeta)\), we write
	\[
		\eta^{-2}L(\zeta)=L(\zeta)+\bigl(\eta^{-2}-1\bigr)L(\zeta).
	\]
	The coefficient \(\bigl(\eta^{-2}-1\bigr)L(\zeta)\) stays bounded as
\(\eta\uparrow1\), and \(K_\eta\) is uniformly Hilbert--Schmidt. Therefore the
extra term $\bigl(\eta^{-2}-1\bigr)L(\zeta)\,K_\eta$
is uniformly Hilbert--Schmidt as well. Collecting all Hilbert--Schmidt pieces
into \(R^{(q)}(\zeta)\) yields
\[
    \widetilde G^{(q)}(\zeta)=L(\zeta)\,K_\eta+R^{(q)}(\zeta),
\]
with \(\sup_{\zeta<\zeta_c}\|R^{(q)}(\zeta)\|_{\HS}<\infty\).

 Finally, each constituent of \(R^{(q)}(\zeta)\) has a pointwise
	limit as \(\zeta\uparrow\zeta_c\), and the dominating Hilbert--Schmidt kernels
	displayed above are independent of \(\zeta\).  Dominated convergence therefore
	yields
	\[
	 R^{(q)}(\zeta)\to R_*^{(q)}
		\qquad\text{in Hilbert--Schmidt norm as }\zeta\uparrow\zeta_c,
	\]
	hence also in operator norm.
\end{proof}

\subsection{From entrywise asymptotics to a rank--one operator}

The prefactor $\eta^{|j_1-j_2|}$ in \eqref{eq:app-entrywise-log} tends to $1$ as
$\zeta\uparrow\zeta_c$.  The next lemma quantifies that, after weighting, it can
be absorbed into the bounded remainder.
\begin{lemma}[Removal of the Toeplitz prefactor]\label{lem:toeplitz-removal}
	Let $D=(\widetilde d^{(q)}_j)_{j\ge0}\in\ell^2(\NN_0)$ and define operators
	\[
		K_\eta:\ (K_\eta)_{j_1j_2}:=\eta^{|j_1-j_2|}D_{j_1}D_{j_2},
		\qquad
		K_1:=D\otimes D^*,
		\qquad 0<\eta\le1.
	\]
	Then \(K_\eta-K_1\to0\) in Hilbert--Schmidt norm as \(\eta\uparrow1\).
	In particular,
	\[
		\|K_\eta-K_1\|\to0,
		\qquad
		L(\zeta)\,\|K_\eta-K_1\|\to0
		\quad\text{as }\zeta\uparrow\zeta_c.
	\]
\end{lemma}

\begin{proof}
	By \eqref{eq:app-d-def}, \eqref{eq:weights} and the explicit constant
	$A_{s,p}=\frac{p\,M^p}{\sqrt{2\pi s(s-1)}}$ in \eqref{eq:app-Raney-asymp}, one has
	\[
		D_j=\widetilde d^{(q)}_j
		=
		c_s\,p_j^{-1-\beta},
		\qquad
		c_s:=\sqrt{\frac{s}{2\pi(s-1)}},
	\]
	hence $a_j:=D_j^2\le C(1+j)^{-2-2\beta}$ and $\sum_{j\ge0}a_j<\infty$.

	Set $b_d:=\sum_{j\ge0}a_j a_{j+d}$ for $d\ge0$. By symmetry,
	\[
		\|K_\eta-K_1\|_{\mathrm{HS}}^2
		=
		\sum_{j_1,j_2\ge0}(1-\eta^{|j_1-j_2|})^2a_{j_1}a_{j_2}
		\asymp
		\sum_{d\ge0}(1-\eta^d)^2 b_d,
	\]
	up to an absolute factor. Moreover, since $a_{j+d}\le C(1+d)^{-2-2\beta}$ for all $j\ge0$,
	\[
		b_d=\sum_{j\ge0}a_j a_{j+d}
		\le
		C(1+d)^{-2-2\beta}\sum_{j\ge0}a_j
		\le
		C'(1+d)^{-2-2\beta}.
	\]

	Fix $N\in\NN$. Using $1-\eta^d\le d(1-\eta)$ for $d\le N$ and $1-\eta^d\le1$ for $d>N$, we get
	\[
		\sum_{d\ge0}(1-\eta^d)^2 b_d
		\le
		(1-\eta)^2\sum_{d=0}^{N} d^2 b_d
		+\sum_{d>N} b_d
		\le
		C'(1-\eta)^2\sum_{d=1}^{N} d^{-2\beta}
		+ C'\sum_{d>N} d^{-2-2\beta}.
	\]
	The last tail is $O(N^{-1-2\beta})$. The partial sum satisfies
	$\sum_{d=1}^{N} d^{-2\beta}=O(N^{1-2\beta})$ if $\beta<\frac12$, $O(\log N)$ if $\beta=\frac12$,
	and $O(1)$ if $\beta>\frac12$. Choose $N:=\lfloor(1-\eta)^{-1}\rfloor$ to balance the two terms. This yields
	\[
		\begin{cases}
			O\!\bigl((1-\eta)^{1+2\beta}\bigr),            & 0<\beta<\frac12, \\[2pt]
			O\!\bigl((1-\eta)^2\log\frac{1}{1-\eta}\bigr), & \beta=\frac12,   \\[2pt]
			O\!\bigl((1-\eta)^2\bigr),                     & \beta>\frac12,
		\end{cases}
		\qquad \eta\uparrow1.
	\]
	In particular,
	\[
		\|K_\eta-K_1\|_{\mathrm{HS}}
		=O\!\bigl((1-\eta)^{\min\{1,\,\frac12+\beta\}}\bigr)
	\]
	up to an extra factor \(\sqrt{\log\frac{1}{1-\eta}}\) when
	\(\beta=\frac12\). Since \(\min\{1,\frac12+\beta\}\ge \min\{1,\beta\}\)
	for all \(\beta>0\), this is the required Hilbert--Schmidt estimate.

	Finally, $\|K_\eta-K_1\|\le\|K_\eta-K_1\|_{\mathrm{HS}}$ and
	$L(\zeta)\sim\log\frac{1}{1-\eta}$ with
	$(1-\eta)^\gamma\log\frac{1}{1-\eta}\to0$ for any $\gamma>0$, yield
	$L(\zeta)\|K_\eta-K_1\|\to0$. This proves the lemma.
\end{proof}


\section{Hypergeometric continuation of \texorpdfstring{$\mathcal G_p$}{Gp}}
\label{app:hyp-2}

This appendix proves the hypergeometric representation used in
Section~\ref{sec:regime-2}. We keep the notation
\(\{\alpha_i\}_{i=1}^{2s}\) and \(\{\beta_j\}_{j=1}^{2s-1}\) from
\eqref{eq:alpha-beta-params}.

\begin{proof}[Proof of Proposition~\ref{prop:Gp-hypergeom}]
	Write
	\[
		\mathcal G_p(u)=\sum_{m\ge0} a_m\,u^m,
		\qquad
		a_m:=R_{s,p}(m)^2,
		\qquad
		R_{s,p}(m)=\frac{p}{sm+p}\binom{sm+p}{m}.
	\]
	Using
	\[
		R_{s,p}(m)
		=
		p\,\frac{\Gamma(sm+p)}{\Gamma(m+1)\Gamma((s-1)m+p+1)},
	\]
	one computes
	\[
		\frac{R_{s,p}(m+1)}{R_{s,p}(m)}
		=
		\frac{\prod_{k=0}^{s-1}(sm+p+k)}
		{(m+1)\prod_{l=1}^{s-1}((s-1)m+p+l)}.
	\]
	Squaring and factorizing the linear terms gives
	\[
		\frac{a_{m+1}}{a_m}
		=
		\zeta_c^{-2}\,
		\frac{\prod_{k=0}^{s-1}\left(m+\frac{p+k}{s}\right)^2}
		{(m+1)^2\prod_{l=1}^{s-1}\left(m+\frac{p+l}{s-1}\right)^2},
	\]
	because
	\[
		\frac{s^{2s}}{(s-1)^{2s-2}}=\zeta_c^{-2}.
	\]
	Now consider the generalized hypergeometric series
	\[
		H(u):={}_{2s}F_{2s-1}\!\left(
		\begin{matrix}
			\alpha_1,\ldots,\alpha_{2s} \\
			\beta_1,\ldots,\beta_{2s-1}
		\end{matrix}
		\Big|\, \zeta_c^{-2}u
		\right),
	\]
	with parameter multisets
	\[
		\{\alpha_i\}
		=
		2\times\left\{\frac{p+k}{s}:k=0,\ldots,s-1\right\},
		\qquad
		\{\beta_j\}
		=
		\{1\}\cup
		2\times\left\{\frac{p+l}{s-1}:l=1,\ldots,s-1\right\}.
	\]
	Its coefficients \(b_m\) satisfy \(b_0=1\) and the standard ratio formula
	\[
		\frac{b_{m+1}}{b_m}
		=
		\zeta_c^{-2}\,
		\frac{\prod_{i=1}^{2s}(m+\alpha_i)}
		{(m+1)\prod_{j=1}^{2s-1}(m+\beta_j)}.
	\]
	By the chosen parameter sets, this is exactly the ratio displayed above for
	\(a_{m+1}/a_m\).  Since \(a_0=R_{s,p}(0)^2=1=b_0\), the sequences coincide
	for all \(m\ge0\), and therefore \(\mathcal G_p(u)=H(u)\). This proves
	\eqref{eq:Gp-hypergeom}.
\end{proof}

\section{The logarithmic coefficient at the branch point}
\label{app:B1}

This appendix computes the coefficient $B(\zeta_c^2)$ in the resonant local
expansion \eqref{eq:resonant-expansion} of $\mathcal G_p$ at the branch point
$u=\zeta_c^2$. The argument isolates the universal $\operatorname{Li}_3$ tail
in the coefficient asymptotics.

\begin{proposition}[Explicit value of $B(\zeta_c^2)$]\label{prop:B1-closed}
	For all integers $s\ge2$ and $p\ge1$, the coefficient $B(\zeta_c^2)$ in
	\eqref{eq:resonant-expansion} satisfies
	\begin{equation}\label{eq:B1-closed}
		B(\zeta_c^2)
		=
		-\frac{p^{2}}{4\pi}\,
		\frac{s^{\,2p-1}}{(s-1)^{\,2p+1}}.
	\end{equation}
	In particular, $B(\zeta_c^2)<0$.
\end{proposition}

\begin{proof}
	Write $\mathcal G_p(u)=\sum_{m\ge0}R_{s,p}(m)^2\,u^m$ and introduce the scaled variable
	$\xi:=u/\zeta_c^2$, so that
	\[
		\mathcal G_p(\zeta_c^2 \xi)=\sum_{m\ge0}a_m \xi^m,
		\qquad
		a_m:=R_{s,p}(m)^2\,\zeta_c^{2m}.
	\]
	By Proposition~\ref{prop:raney-asymp}, one has the refined asymptotic
	\begin{equation}\label{eq:am-refined}
		a_m = C_{s,p}\,m^{-3}+O(m^{-4}),
		\qquad
		C_{s,p}:=\frac{p^{2}}{2\pi}\,\frac{s^{\,2p-1}}{(s-1)^{\,2p+1}}.
	\end{equation}
	On the other hand, the polylogarithm $\mathrm{Li}_3(\xi)=\sum_{m\ge1}m^{-3}\xi^m$
	has the classical singular expansion at $\xi=1$ (see, e.g., DLMF~\S25.12
	\cite{NIST:DLMF})
	\[
		\mathrm{Li}_3(\xi)
		= \text{(analytic at $\xi=1$)} - \tfrac12(1-\xi)^2\log(1-\xi) + O\!\bigl((1-\xi)^2\bigr),
		\qquad \xi\to1.
	\]

	Define the remainder series
	\[
		H(\xi):=\mathcal G_p(\zeta_c^2\xi)-C_{s,p}\,\mathrm{Li}_3(\xi)
		=\sum_{m\ge1}\Bigl(a_m-C_{s,p}m^{-3}\Bigr)\xi^m.
	\]
	By \eqref{eq:am-refined}, the coefficients satisfy $a_m-C_{s,p}m^{-3}=O(m^{-4})$,
	hence $\sum_{m\ge1}m^2\bigl|a_m-C_{s,p}m^{-3}\bigr|<\infty$. Therefore the series for
	$H(\xi)$, $H'(\xi)$ and $H''(\xi)$ converge absolutely at $\xi=1$, so $H$ extends to a
	$C^2$-function on $[0,1]$ and therefore
	\[
		H(\xi)=H(1)+H'(1)(\xi-1)+\tfrac12H''(1)(\xi-1)^2+o\bigl((1-\xi)^2\bigr),
		\qquad \xi\to1.
	\]
	In particular, $H$ contributes no $(1-\xi)^2\log(1-\xi)$ term.

	Consequently, the $(1-\xi)^2\log(1-\xi)$ coefficient in $\mathcal G_p(\zeta_c^2\xi)$ is exactly
	$-\tfrac12C_{s,p}$, coming from $C_{s,p}\mathrm{Li}_3(\xi)$. Comparing with
	\eqref{eq:resonant-expansion} yields $B(\zeta_c^2)=-\tfrac12C_{s,p}$, which is
	\eqref{eq:B1-closed}.
\end{proof}

\begin{corollary}[Rational--over--$\pi$ structure]\label{cor:B1-rational}
	For all integers $s\ge2$ and $p\ge1$ one has $\pi\,B(\zeta_c^2)\in\QQ$ and $\pi\,B(\zeta_c^2)<0$.
	Consequently, the edge value in Theorem~\ref{thm:edge-density} satisfies
	$\rho_p(\zeta_c^2)\in\pi^{-1}\QQ$.
\end{corollary}

\begin{proof}
	The first claim is immediate from the closed formula \eqref{eq:B1-closed}. The
	expression for $\rho_p(\zeta_c^2)$ from Theorem~\ref{thm:edge-density} is a
	rational multiple of $-B(\zeta_c^2)$, so the same rational-over-$\pi$
	structure is inherited by the edge value.
\end{proof}

\end{document}